\documentclass[a4paper,UKenglish,cleveref,autoref]{lipics-v2021}
\hideLIPIcs

\usepackage{mathtools}
\usepackage{booktabs}
\usepackage{xspace}
\usepackage{tikz}
\usepackage{cite}
\usetikzlibrary{arrows.meta,calc,positioning,decorations.pathreplacing}

\newcommand{\AP}{\mathsf{AP}}
\newcommand{\Ag}{\mathsf{Ag}}
\newcommand{\Nat}{\mathbb{N}}

\newcommand{\KMTL}{\mathrm{KMTL}}
\newcommand{\KMTLone}{\mathrm{KMTL}_1}
\newcommand{\MTL}{\mathrm{MTL}^1}
\newcommand{\obs}{\mathsf{obs}}
\newcommand{\diag}{\mathsf{diag}}
\newcommand{\pref}[2]{#1_{\le #2}}

\newcommand{\M}{\mathcal{M}}
\newcommand{\Cell}{\mathsf{Cell}}
\newcommand{\End}{\mathsf{End}}
\newcommand{\Main}{\mathsf{Main}}
\newcommand{\Lo}{\mathsf{L}}
\newcommand{\Up}{\mathsf{U}}
\newcommand{\SameL}{\mathsf{SameL}}
\newcommand{\Pair}{\mathsf{Pair}}
\newcommand{\Good}{\mathsf{Good}}
\newcommand{\Tile}{\mathsf{Tile}}
\newcommand{\Bit}{\mathsf{Bit}}
\newcommand{\Zero}{\mathsf{Zero}}
\newcommand{\Max}{\mathsf{Max}}
\newcommand{\Carry}{\mathsf{Carry}}
\newcommand{\Inc}{\mathsf{Inc}}
\newcommand{\Base}{\Phi_{\mathsf{base}}}
\newcommand{\Col}{\Phi_{\mathsf{col}}}
\newcommand{\PhiI}{\Phi_{\mathcal{T}}}
\newcommand{\NextCell}{\mathsf{NextCell}}
\newcommand{\Pin}{P^{\mathsf{in}}}
\newcommand{\EXPSPACE}{\ensuremath{\mathsf{EXPSPACE}}\xspace}
\newcommand{\PSPACE}{\ensuremath{\mathsf{PSPACE}}\xspace}

\title{Agent-Alternation-Free Epistemic Metric Temporal Logic with Past: Model Checking and Complexity}
\titlerunning{Epistemic Metric Temporal Logic with Past}

\author{Benedikt Bollig}%
  {Université Paris-Saclay, CNRS, ENS Paris-Saclay, LMF, Gif-sur-Yvette, France}%
  {bollig@lmf.cnrs.fr}%
  {https://orcid.org/0000-0003-0985-6115}%
  {}

\author{Matthias F{\"u}gger}%
  {Université Paris-Saclay, CNRS, ENS Paris-Saclay, LMF, Gif-sur-Yvette, France}%
  {mfuegger@lmf.cnrs.fr}%
  {https://orcid.org/0000-0001-5765-0301}%
  {}

\author{Thomas Nowak}%
  {Université Paris-Saclay, CNRS, ENS Paris-Saclay, LMF, Gif-sur-Yvette, France \and
   Institut Universitaire de France, Paris, France}%
  {thomas@thomasnowak.net}%
  {https://orcid.org/0000-0003-1690-9342}%
  {}

\author{Paul Zeinaty}%
  {Université Paris-Saclay, CNRS, ENS Paris-Saclay, LMF, Gif-sur-Yvette, France \and
   Direction Générale de l'Armement, Paris, France}%
  {pzeinaty@lmf.cnrs.fr}%
  {https://orcid.org/0009-0003-2885-0410}%
  {}

\authorrunning{B.~Bollig, M.~F{\"u}gger, T.~Nowak, and P.~Zeinaty}
\Copyright{Benedikt Bollig, Matthias F{\"u}gger, Thomas Nowak, and Paul Zeinaty}



\funding{This work was partly supported by the ``France 2030'' government
investment plan managed by ANR, under the reference ANR-23-PEIA-0006.}

\nolinenumbers

\relatedversiondetails{Conference Version}{https://doi.org/10.4230/LIPIcs.CSL.2027.27}

\EventEditors{Corina C\^{\i}rstea and Yannick Forster}
\EventNoEds{2}
\EventLongTitle{35th EACSL Annual Conference on Computer Science Logic (CSL 2027)}
\EventShortTitle{CSL 2027}
\EventAcronym{CSL}
\EventYear{2027}
\EventDate{January 26--29, 2027}
\EventLocation{Brighton, UK}
\EventLogo{}
\SeriesVolume{403}
\ArticleNo{27}

\begin{document}

\maketitle

\begin{abstract}
We study model checking for an epistemic metric temporal logic with
past, interpreted over finite B{\"u}chi automata under synchronous perfect
recall. The logic is motivated by observation-based verification problems such
as diagnosis and opacity, where an observer sees only a projection of an
execution and reasons about events that may have occurred earlier. These
requirements use no alternation between different agents' knowledge. We
therefore consider the agent-alternation-free fragment, in which nested
knowledge operators must refer to the same agent.
We show that model checking for this
fragment is \EXPSPACE-complete. The lower bound already holds with one agent,
one occurrence of the knowledge operator, and no non-trivial metric bounds. For
the upper bound, we combine temporal test automata with perfect-recall
observers. Because past formulas may have different truth values on
indistinguishable histories ending in the same system state, the observer must
track temporal automaton states in addition to system states.
\end{abstract}

\section{Introduction}
\label{sec:introduction}

Many verification questions concern what an observer can infer from a partial
observation of an execution. This
perspective is standard in discrete-event systems. Diagnosability asks whether
an observer can detect an unobservable fault after a bounded delay~\cite{SampathSLST95}.
Decentralized diagnosis considers several local observers with different projections of the same
execution~\cite{DeboukLT00}. Opacity imposes the dual confidentiality
requirement: the observations should not reveal that a secret event has
occurred~\cite{BryansKMR08,CassezDM09,LafortuneLH18}.

A related pattern appears in runtime verification~\cite{LeuckerS09,BauerLS11}.
A monitor reads a finite observation prefix and produces a verdict only when
it is justified by every execution compatible with the observations. For
example, a fault alarm is justified once all compatible executions contain a
fault, even if the fault itself was invisible. We capture this reasoning by
knowledge under synchronous perfect recall~\cite{FHMV1995}: an observer
remembers the complete sequence of its observations, including the number of
elapsed steps, and considers precisely the system executions with the same
observation history. Temporal operators then describe what must hold on these
executions.

This gives a common logical formulation of diagnosis and opacity. A diagnosis
requirement states that the observer eventually knows that a fault occurred.
An opacity requirement states that it never knows that a secret occurred. For
a given system, both become model-checking questions for epistemic temporal
formulas.

Past operators are essential for these applications. A diagnoser must know
that a fault occurred at some earlier position. Likewise, opacity commonly protects the occurrence
of a past secret. Metric constraints (encoded in binary)
express timing requirements, such as
diagnosis within \(D \in \Nat\) steps or secrecy for a prescribed window. We work in
discrete time, with distances counted in steps. This restriction is important:
in dense time, opacity is already undecidable~\cite{Cassez09},
by reduction from universality of timed automata~\cite{AlurD94}.
Our logic therefore combines future and past temporal operators with strict
first-time semantics,
integer metric bounds, and synchronous perfect-recall knowledge. Under strict
first-time semantics, a future metric operator constrains the distance to the
first position satisfying its argument. Thus, requiring the first alarm to
occur between three and five steps from now excludes an alarm after one step,
even if another occurs after four steps. A past metric operator refers,
symmetrically, to the most recent occurrence. Satisfiability of the future
knowledge-free fragment is \PSPACE-complete~\cite{Alsmann025}, whereas
standard metric operators with two-sided intervals yield
\EXPSPACE-complete model checking~\cite{LaroussinieMS06}.

The complexity of perfect-recall knowledge is sensitive to the
temporal operators and to agent alternation. For epistemic LTL with future
operators and no metric constraints, model checking is nonelementary in
general and \PSPACE-complete at knowledge-alternation depth
one~\cite{BozzelliMM24}. The pure past fragment, without future or metric
operators, has also been studied under synchronous perfect
recall~\cite{CohenL10}. Combining future and past leads to the logic KLTL,
which forms the basis of a richer counterfactual logic~\cite{FinkbeinerS25}
that is decidable over finite-state systems but whose exact complexity is
left open. These results leave open the cost of adding past to the
agent-alternation-free future fragment.

The FDI framework of~\cite{BozzanoCGT15} uses past and knowledge to express
bounded-delay alarm requirements, representing delay bounds by finite
unfoldings of temporal operators. Metric temporal operators are handled
directly by bounded model-checking approaches for timed interpreted
systems~\cite{LomuscioPW07,Wozna-Szczesniak16}, where knowledge is interpreted
from the current local state and only future temporal operators are available.

In this paper, we consider future and past temporal operators with binary-encoded metric
bounds and synchronous perfect-recall knowledge, and prove that model checking is \EXPSPACE-complete at
knowledge-alternation depth one. The lower bound already holds for the metric-free fragment. Thus
\EXPSPACE-hardness appears when past is added to the \PSPACE-complete
metric-free future epistemic fragment.
A related succinctness phenomenon is known without knowledge:
the work \cite{LMS2002} shows that adding a forgettable-past
operator to LTL with past yields \EXPSPACE-complete satisfiability and model
checking. Our lower bound adapts its tiling reduction.

Our contributions are as follows:
\begin{enumerate}[(1)]
\item We express positive and negative diagnosis, their decentralized variants,
      and opacity in the agent-alternation-free fragment $\KMTLone$.
\item We give an \EXPSPACE lower bound by reduction from exponential-width
      corridor tiling. The formula has one occurrence of a knowledge operator.
\item We give a matching upper bound by constructing an observer for each
      knowledge subformula. Each observer records sets of pairs of a
      system state and a temporal automaton state, distinguishing histories that
      end in the same system state but disagree on past formulas. For nested
      knowledge, accepting continuations also carry the observers for inner
      knowledge subformulas, so that these subformulas can be evaluated along
      the continued observation history.
\end{enumerate}

\paragraph*{Outline}
\Cref{sec:syntax-semantics} introduces the logic $\KMTL$ and its fragment $\KMTLone$.
\Cref{sec:examples} relates it to diagnosability and opacity.
\Cref{sec:lower-bound} proves the \EXPSPACE lower bound and
\Cref{sec:upper-bound} the matching upper bound.
We conclude in \Cref{sec:conclusion}.

\section{The Logic \texorpdfstring{$\KMTL$}{KMTL} and its Fragment \texorpdfstring{$\KMTLone$}{KMTL\_1}}
\label{sec:syntax-semantics}

The syntax and semantics are parameterized by finite sets $\AP$ of atomic
propositions and $\Ag$ of agents. Both are part of the model-checking instance
(Definition~\ref{def:mc}). We write $\Sigma=2^{\AP}$, viewing a letter
$\sigma\in\Sigma$ as the set of propositions true at one position. Positions
are numbered from~$1$. Thus an infinite word is written
$w=\sigma_1\sigma_2\ldots\in\Sigma^\omega$. For $i\in\Nat$, let
$\pref{w}{i}=\sigma_1\ldots\sigma_i\in\Sigma^*$, with $\pref{w}{0}$ the empty
word. A \emph{pointed word} is a pair $(w,i)$ with $i\ge1$. Temporal operators
move the distinguished position along $w$, while knowledge compares $(w,i)$
with pointed words having the same observation history.

Formulas are interpreted over the language of a B\"uchi automaton:

\begin{definition}[B\"uchi automata]
\label{def:buchi-automata}
A (nondeterministic finite) B\"uchi automaton over
\(\Sigma\) is a tuple \[A=(S,S_0,\delta,F)\]
where \(S\) is a finite set of states, \(S_0\subseteq S\) is the set of initial
states, \(\delta\subseteq S\times\Sigma\times S\) is the transition relation,
and \(F\subseteq S\) is the set of accepting states.  A run over
\(w=\sigma_1\sigma_2\ldots\in\Sigma^\omega\) is an infinite sequence
\(s_0s_1s_2\ldots\) with \(s_0\in S_0\) and
\((s_i,\sigma_{i+1},s_{i+1})\in\delta\) for all \(i\ge 0\).  It is accepting if
\(s_i\in F\) for infinitely many \(i\).  The language \(L(A)\) is the set of
words admitting an accepting run.  A \emph{run prefix} over a finite word
\(\sigma_1\ldots\sigma_m\in\Sigma^*\) is a finite sequence \(s_0\ldots s_m\) with
\(s_0\in S_0\) and \((s_i,\sigma_{i+1},s_{i+1})\in\delta\) for \(0\le i<m\), also
written \[s_0\xrightarrow{\sigma_1}\cdots\xrightarrow{\sigma_m}s_m.\]

A generalized B\"uchi automaton has a finite family
\(\mathcal F=\{F_1,\ldots,F_m\}\) of accepting sets instead of one set \(F\).
A run is accepting if, for every \(j\), it visits \(F_j\) infinitely often.
\end{definition}

A model is a B\"uchi automaton over $\Sigma=2^{\AP}$. We write it as
$\M=(Q,Q_0,\Delta,F_\M)$ and reserve $S,S_0,\delta,F$ for automata constructed
in the proofs. Formulas are interpreted over $L(\M)\subseteq\Sigma^\omega$, and
the finite transition structure of $\M$ is the input representation used by
the algorithms. In Section~\ref{sec:upper-bound}, we construct temporal test
automata to check whether a knowledge-free formula holds at a given position
of a word. They use generalized B\"uchi acceptance over $2^{\AP'}$, where
$\AP'$ is a finite set of propositions that may include fresh propositions
introduced by the construction.

\begin{definition}[Observations and synchronous perfect recall]
Each agent $a\in\Ag$ comes with a set $\AP_a\subseteq\AP$ of propositions visible
to $a$.  The observation of a letter $\sigma\in2^{\AP}$ is
\[
  \obs_a(\sigma)=\sigma\cap \AP_a.
\]
For a finite word $u=\sigma_1\ldots\sigma_i\in\Sigma^*$, define
\(
  \obs_a(u)=\obs_a(\sigma_1)\ldots\obs_a(\sigma_i).
\)
Two pointed words $(w,i)$ and $(w',i')$ are synchronously indistinguishable for
$a$, written $(w,i)\sim_a(w',i')$, if we have
\[
  i=i' \quad\text{and}\quad
  \obs_a(\pref{w}{i})=\obs_a(\pref{w'}{i}).
\]
\end{definition}

Metric constraints use nonempty intervals $I\subseteq\Nat$ of the form
$[\ell,h]=\{n\in\Nat\mid \ell\le n\le h\}$ or
$[\ell,\infty)=\{n\in\Nat\mid \ell\le n\}$, where $\ell,h\in\Nat$.
Finite bounds are encoded in binary. We write $\Nat=[0,\infty)$ for the
unbounded interval.

\begin{definition}[$\KMTL$ formulas]
The formulas of $\KMTL$ are generated by
\[
\varphi ::= \top \mid p \mid \neg\varphi \mid \varphi\wedge\varphi
\mid X\varphi \mid Y\varphi
\mid \varphi\,U^1_I\,\varphi
\mid \varphi\,S^1_I\,\varphi
\mid K_a\varphi,
\]
where $p\in\AP$, $a\in\Ag$, and $I$ is a nonempty interval as defined above.
The classical Boolean connectives $\vee$, $\rightarrow$, and $\leftrightarrow$
are defined as abbreviations.
\end{definition}

Following~\cite{Alsmann025}, the superscript $1$ denotes strict first-time
semantics. The formula $\varphi\,U^1_I\,\psi$ requires the first
$\psi$-position at or after the current position to occur at a distance in $I$,
with $\varphi$ holding before it. Symmetrically,
$\varphi\,S^1_I\,\psi$ refers to the most recent $\psi$-position. In terms of
classical metric until, $\varphi\,U^1_I\,\psi$ is equivalent to
$(\varphi\wedge\neg\psi)\,U_I\,\psi$. The formal semantics is as follows.

\begin{definition}[Satisfaction]
Fix a B\"uchi automaton \(\M\).  Satisfaction
\(w,i\models_\M\varphi\) at a pointed word \((w,i)\), for \(w=\sigma_1\sigma_2 \ldots\in L(\M)\) and
\(i\ge 1\), is defined inductively as follows:
\[
\begin{array}{lcl}
w,i\models_\M \top && \text{always},\\[0.5ex]
w,i\models_\M p &\text{if}& p\in\sigma_i,\\[0.5ex]
w,i\models_\M \neg\varphi &\text{if}& w,i\not\models_\M\varphi,\\[0.5ex]
w,i\models_\M \varphi\wedge\psi &\text{if}& w,i\models_\M\varphi\text{ and }w,i\models_\M\psi,\\[0.5ex]
w,i\models_\M X\varphi &\text{if}& w,i+1\models_\M\varphi,\\[0.5ex]
w,i\models_\M Y\varphi &\text{if}& i>1\text{ and }w,i-1\models_\M\varphi,\\[0.5ex]
w,i\models_\M \varphi\,U^1_I\,\psi &\text{if}&
  \text{there is }j\ge i\text{ with }j-i\in I,\ w,j\models_\M\psi,\text{ and}\\
 && w,k\models_\M \varphi\wedge\neg\psi\text{ for all }i\le k<j,\\[0.5ex]
w,i\models_\M \varphi\,S^1_I\,\psi &\text{if}&
  \text{there is }j\le i\text{ with }i-j\in I,\ w,j\models_\M\psi,\text{ and}\\
 && w,k\models_\M \varphi\wedge\neg\psi\text{ for all }j<k\le i,\\[0.5ex]
w,i\models_\M K_a\varphi &\text{if}&
  w',i\models_\M\varphi\text{ for all }w'\in L(\M) \text{ with }(w,i)\sim_a(w',i).
\end{array}
\]
\end{definition}

We also use the following abbreviations, for any interval $I$:
\[
\begin{array}{rclcrclcrcl}
\Diamond_I\varphi &:=& \top\,U^1_I\,\varphi, &&
\Diamond\varphi &:=& \Diamond_{\Nat}\varphi, &&
\Box\varphi &:=& \neg\Diamond\neg\varphi,\\[2mm]
\Diamond^{-}_I\varphi &:=& \top\,S^1_I\,\varphi, &&
\Diamond^{-}\varphi &:=& \Diamond^{-}_{\Nat}\varphi, &&
\Box^{-}\varphi &:=& \neg\Diamond^{-}\neg\varphi.
\end{array}
\]
Note that $\Diamond_I\varphi$ says that the \emph{first}
position satisfying $\varphi$ lies in the window $I$. For $I=\Nat$ this is the
usual eventuality, and $\Box$ and $\Box^{-}$ are the usual universal modalities.
For $k\in\Nat$, $X^k\varphi$ denotes $k$ iterated applications of $X$, with
$X^0\varphi=\varphi$.

We restrict alternation between knowledge operators of different agents.

\begin{definition}[Agent-alternation-free fragment $\KMTLone$]
\label{def:fragment}
A $\KMTL$ formula is \emph{agent-alternation-free} if along every chain of nested
knowledge operators the agent does not change.  We
write $\KMTLone$ for the set of agent-alternation-free $\KMTL$ formulas.
\end{definition}

Equivalently, no $K_b$ with $b\ne a$ may occur in the scope of $K_a$.
Same-agent nesting is allowed, and non-nested knowledge subformulas may refer to
different agents. This is knowledge-alternation depth at most one in the sense
of~\cite{BozzelliMM24}. The subscript in $\KMTLone$ records this depth. By
contrast, the superscript in the metric logic $\MTL$ of~\cite{Alsmann025} refers
to first-time semantics, which we mark directly on the temporal operators.

\begin{definition}[Model checking]
\label{def:mc}
An instance of the (universal) model-checking problem consists of finite sets $\AP$
and $\Ag$, observation sets $(\AP_a)_{a\in\Ag}$, a B\"uchi automaton \(\M\)
over the alphabet $\Sigma=2^{\AP}$, and a formula \(\varphi\) over $\AP$ and
$\Ag$.  It asks whether
\[
  \M\models\varphi
  \quad\text{meaning}\quad
  \forall w\in L(\M):~ w,1\models_\M\varphi.
\]
\end{definition}

\begin{remark}[State-based observations and B\"uchi acceptance]
\label{rem:state-based}
The paper~\cite{BozzelliMM24} uses state-based observations
on Kripke structures and quantifies over all paths. Our observations are
letter-based, and knowledge ranges over the accepted language $L(\M)$. These
choices do not affect the cited \PSPACE comparison.

State-based and letter-based observations are polynomially interreducible. To
encode a state equivalence $\sim_a$ by letters, use at most
$\lceil\log_2|Q|\rceil$ propositions visible to $a$ that record the current
class. Conversely, store the
last emitted letter in the state and let $\sim_a$ compare its
$\AP_a$-projection.

B\"uchi acceptance can be relativized in the formula. Let
$\mathsf{Fair}:=\Box\Diamond\mathsf{acc}$, where the fresh proposition
$\mathsf{acc}$ marks accepting states and is visible to no agent. Translate
knowledge recursively by
$K_a\psi\mapsto K_a(\mathsf{Fair}\rightarrow\psi')$ and use
$\mathsf{Fair}\rightarrow\varphi'$ at the top level. The resulting formula
considers only accepting paths, both globally and inside knowledge operators,
and preserves agent alternation depth.
\end{remark}

\section{Motivating Examples from Diagnosis and Opacity}
\label{sec:examples}

Diagnosis and opacity are defined through partial observations. At a given
position, an observer must consider every execution compatible with the
observation prefix. Knowledge captures this quantification, while past
operators describe the event to be inferred, such as an earlier fault or
secret.

Classical definitions are usually phrased in terms of diagnosers and
indistinguishable runs. Diagnosis has also been formalized in temporal epistemic
logic, notably in the FDI framework of~\cite{BozzanoCGT15}, and opacity has been
expressed epistemically as $\Box\neg K\,\mathit{sec}$~\cite{MaubertPB2011}, where
$\mathit{sec}$ represents the secret. We show below that positive and negative
diagnosis, their decentralized variants, and past-event opacity are all
expressed by formulas with past and no nested knowledge operators.

\paragraph*{Bounded diagnosis}
Let $a \in \Ag$ be an observing agent and let \(f\in\AP\setminus\AP_a\) be an unobservable proposition marking a fault occurrence.
A finite or infinite word is \emph{fault-free} if none of its letters contains $f$.
For a trace \(w\) and a position \(i\), the formula \(\Diamond^- f\) says that a
fault occurred in the prefix \(w_{\le i}\).
Let $D \in \Nat$ be a delay and write $(2^{\AP_a})^+$ for the set of nonempty
finite observation words. Following the classical diagnoser formulation
of~\cite{SampathSLST95}, we call a function
\[
  \diag:(2^{\AP_a})^+\to\{0,1\}
\]
a \(D\)-delay diagnoser for \(\M\) if:
\begin{enumerate}[(i)]
\item if a nonempty finite prefix \(u\) of some \(w\in L(\M)\) is fault-free, then
      \(\diag(\obs_a(u))=0\).
\item if \(w\in L(\M)\) has a fault at position \(i\), then
      \(\diag(\obs_a(w_{\le j}))=1\) for all \(j\ge i+D\).
\end{enumerate}
The first rule forbids false alarms. The second requires the alarm to stay on
from \(D\) steps after every fault onward, matching the uniform-delay condition
of~\cite{SampathSLST95}.

\begin{lemma}[Bounded diagnosis as knowledge]
\label{lem:diagnosis-equivalence}
For every \(D\in\Nat\), the system \(\M\) has a \(D\)-delay diagnoser for
observer \(a\) and fault \(f\) iff \(\M\models\varphi_{\mathsf{diag}}^{\le D}\), where
\[
  \varphi_{\mathsf{diag}}^{\le D} := \Box\bigl(f\rightarrow \Diamond_{[0,D]} K_a\Diamond^- f\bigr).
\]
\end{lemma}

\begin{proof}
Write \(\alpha:=K_a\Diamond^- f\) for the alarm predicate.  It depends only on the
observation prefix, since indistinguishable positions have the same \(a\)-class,
and it is monotone under perfect recall (once the fault is known it stays known).
Since \(\Diamond_{[0,D]}\alpha\) requires \(\alpha\) at some position within \(D\)
steps and \(\alpha\) persists once true, \(\varphi_{\mathsf{diag}}^{\le D}\) requires
the alarm to hold from position \(i+D\) onward after any fault at \(i\).

Suppose the formula holds. Set \(\diag(u)=1\) iff \(\alpha\) holds after some
(equivalently every) trace prefix with observation \(u\), and \(\diag(u)=0\) for
observations not realized by any trace prefix.  A fault-free prefix is an
indistinguishable non-faulty witness, so \(\alpha\) is false there and there are no
false alarms.  After a fault at \(i\) the formula gives \(\alpha\) from \(i+D\)
onward, hence \(\diag(\obs_a(w_{\le j}))=1\) for all \(j\ge i+D\).

Conversely, let \(\diag\) be a \(D\)-delay diagnoser and let \(f\) occur at
position \(i\) of \(w\). Detection by the diagnoser gives \(\diag(\obs_a(w_{\le i+D}))=1\), so every
trace prefix \(w'_{\le i+D}\) with the same observation is faulty (otherwise the
no-false-alarm rule would force \(\diag=0\)).  Thus \(w,i+D\models\alpha\), which
lies in the window \([0,D]\), so \(\varphi_{\mathsf{diag}}^{\le D}\) holds.
\end{proof}

The unbounded eventual-diagnosis formula
\(
  \varphi_{\mathsf{diag}}
  :=
  \Box\bigl(f\rightarrow \Diamond K_a\Diamond^- f\bigr)
\)
drops the delay bound and captures the qualitative requirement.  The bound can
also be written with iterated next as \(\Box(f\rightarrow X^D K_a\Diamond^- f)\).
By monotonicity of \(K_a\Diamond^- f\), this is equivalent to
\(\varphi_{\mathsf{diag}}^{\le D}\), but it expresses the delay \(D\) in unary,
whereas the metric constant in \(\Diamond_{[0,D]}\) is binary.

\begin{example}[Diagnosis, perfect recall, and delay]
\label{ex:diagnosis}
Let $\AP=\{p,r,f\}$ and $\AP_a=\{p,r\}$, so the fault $f$ is unobservable.
Consider three systems whose trace languages are given by the following
$\omega$-regular expressions over $\Sigma=2^{\AP}$, in which each set, such as
$\{p\}$ or $\{p,f\}$, is a single letter:
\[
\begin{aligned}
L(\M_1) &= \{p\}^\omega + \{p\}^{*}\{p,f\}\{p\}^\omega,\\
L(\M_2) &= \{p\}^\omega + \{p\}^{*}\{f\}\{p\}^\omega,\\
L(\M_3) &= \{p\}^\omega + \{p\}^{*}\{p,f\}\{p\}^{*}\{r\}^\omega.
\end{aligned}
\]
Agent $a$ observes $\{f\}$ as $\emptyset$ and $\{p,f\}$ as $\{p\}$.
In $\M_1$, faulty and fault-free traces have the same observation
$\{p\}^\omega$, so the fault is never known:
$\M_1\not\models\varphi_{\mathsf{diag}}$. In $\M_2$, the fault position is
observed as $\emptyset$, which cannot occur on the fault-free trace. Hence the
fault is known immediately and
$\M_2\models\varphi_{\mathsf{diag}}^{\le 0}$. Perfect recall preserves this
knowledge after the observation returns to $\{p\}$. In $\M_3$, the fault becomes
known only when the observation switches to $\{r\}^\omega$. The switch occurs
on every faulty trace, but after an
arbitrarily long prefix indistinguishable from the fault-free trace. Therefore
$\M_3\models\varphi_{\mathsf{diag}}$, while
$\M_3\not\models\varphi_{\mathsf{diag}}^{\le D}$ for every $D$.
\end{example}

\paragraph*{Negative diagnosis}
An observer may also need to certify the absence of faults in an earlier
prefix. For example, a monitor of a sampled device may guarantee that, whenever
the device has remained fault-free, its history up to ten samples ago can be
certified as fault-free. This is the purpose of negative
diagnosability~\cite[Definition~3]{WangYL07}. We use its synchronous version:
for every \(w\in L(\M)\) and \(i>D\), if \(w_{\le i}\) is fault-free, then
\(w'_{\le i-D}\) is fault-free for every \(w'\in L(\M)\) with
\((w,i)\sim_a(w',i)\).
In~\cite{WangYL07}, observations erase unobservable events and compatible
prefixes may have different lengths. Our observations retain every position,
so all compatible prefixes have length \(i\). Moreover, the literature asks
for the existence of a uniform delay bound, whereas we fix \(D\).

\begin{lemma}[Negative diagnosis as knowledge]
\label{lem:negative-diagnosis-equivalence}
For every \(D\in\Nat\), the model \(\M\) satisfies the synchronous
negative-diagnosis condition with delay \(D\) for observer \(a\) and fault \(f\)
iff \(\M\models\varphi_{\mathsf{ndiag}}^{\le D}\), where
\begin{equation}
\label{eq:negative-diagnosis}
  \varphi_{\mathsf{ndiag}}^{\le D}:=
  \Box\left(
    \neg\Diamond^- f
    \rightarrow K_a\neg\Diamond^-_{[D,\infty)}
      \bigl(f\wedge\neg Y\Diamond^- f\bigr)
  \right).
\end{equation}
\end{lemma}

\begin{proof}
The formula \(f\wedge\neg Y\Diamond^- f\) holds exactly at the first fault
position, if a fault occurs. Since this position is unique,
\(\Diamond^-_{[D,\infty)}(f\wedge\neg Y\Diamond^- f)\) holds at position
\(i\) iff some fault occurred at a position \(j\) with \(j+D\le i\).
Its negation is therefore automatic when \(i\le D\), and otherwise says
that the prefix ending at \(i-D\) is fault-free. The antecedent says that
the actual prefix ending at \(i\) is fault-free, while knowledge requires
the earlier-prefix condition on every accepted trace with the same
observations. This is exactly the stated condition, including when \(D=0\).
\end{proof}

Equivalently, we could write
\(X^D\Box(\neg\Diamond^- f\rightarrow K_aY^D\neg\Diamond^- f)\), where \(Y^D\)
denotes \(D\) iterated applications of \(Y\), with \(Y^0\psi=\psi\).
Formula~\eqref{eq:negative-diagnosis} encodes the delay in binary.

For one observer, positive and negative diagnosability are equivalent for the
same bound \(D\). Both fail precisely when a fault-free prefix has the same
observation as a prefix containing a fault at least \(D\) steps earlier.
For positive diagnosis, such prefixes can be truncated to exactly \(D\) steps
after that fault. Perfect recall preserves their indistinguishability under
truncation. This also explains the centralized equivalence
in~\cite[Theorem~2]{WangYL07} in our synchronous setting.

\paragraph*{Decentralized diagnosis}
The translation extends to decentralized diagnosis. Suppose several agents
observe different projections, $f\notin\AP_a$ for every $a\in\Ag$, and no fixed
agent is required to diagnose every fault. Under the codiagnosability notion
of~\cite{QiuK06}, implemented by the disjunctive protocol
of~\cite{DeboukLT00}, a fault is detected once at least one agent can infer it
from its local observations.
A \(D\)-delay decentralized diagnoser is a family
\((\diag_a)_{a\in\Ag}\) of local alarm functions
\(\diag_a:(2^{\AP_a})^+\to\{0,1\}\), each subject to the no-false-alarm
rule~(i) above for its own observation alphabet, such that whenever
\(w\in L(\M)\) has a fault at position \(i\), then for every \(j\ge i+D\) at
least one \(a\in\Ag\) has \(\diag_a(\obs_a(w_{\le j}))=1\).

\begin{lemma}[Decentralized diagnosis as knowledge]
\label{lem:codiagnosis-equivalence}
For every \(D\in\Nat\), the model \(\M\) has a \(D\)-delay decentralized diagnoser
for fault \(f\) iff
\[
  \M\models
  \Box\left(
    f\rightarrow
    \Diamond_{[0,D]}\bigvee_{a\in\Ag}K_a\Diamond^- f
  \right).
\]
\end{lemma}

\begin{proof}
The argument is the same as that of \Cref{lem:diagnosis-equivalence}, now with one alarm
\(\diag_a\) per agent, firing on the prefixes where \(K_a\Diamond^-f\) holds.  The
disjunction provides detection by at least one agent within \(D\) steps, monotonicity
keeps that agent's alarm on, and the no-false-alarm requirement rules out an alarm along a
fault-free prefix.
\end{proof}

This formula uses several knowledge operators, but they are not nested.  Hence
decentralized diagnosis remains in our fragment $\KMTLone$.

For a nonempty family of observers, negative
codiagnosability~\cite[Definition~7]{WangYL07} requires that every fault-free
prefix allow at least one observer to certify its earlier history.
Its synchronous version is obtained from~\eqref{eq:negative-diagnosis} by
replacing \(K_a\psi\) with \(\bigvee_{a\in\Ag}K_a\psi\).
This formula also belongs to $\KMTLone$.

\begin{example}[Negative codiagnosability]
\label{ex:negative-codiagnosis}
Consider a device whose startup test either succeeds and sends a report on
one of two local channels, or fails without sending a report. The device then
remains idle. Observer \(a_j\) receives only the reports on channel \(j\).
Let \(\Ag=\{a_1,a_2\}\), \(\AP=\{p_1,p_2,f\}\), and
\(\AP_{a_j}=\{p_j\}\) for \(j\in\{1,2\}\). Consider the system with
\[
  L(\M)=\bigl(\{p_1\}+\{p_2\}+\{f\}\bigr)\emptyset^\omega.
\]
Here, \(p_j\) records success reported on channel \(j\), and \(f\) records
the hidden startup failure.
On a fault-free trace, the observer who sees \(p_j\) knows immediately that
no fault has occurred. Perfect recall preserves this knowledge, so the system
is negatively codiagnosable with delay~\(0\).
On the faulty trace \(\{f\}\emptyset^\omega\), observer \(a_1\) cannot exclude
\(\{p_2\}\emptyset^\omega\), and observer \(a_2\) cannot exclude
\(\{p_1\}\emptyset^\omega\). Thus neither observer ever knows that a fault
occurred, and positive codiagnosability fails for every delay bound.
This is the synchronous counterpart of~\cite[Example~2]{WangYL07}.
\end{example}

\paragraph*{Opacity}
Let \(\mathit{sec}\in\AP\setminus\AP_a\) be an unobservable secret event.  The standard
trace view of opacity says that every secret observation prefix must also be
compatible with a non-secret prefix~\cite{BryansKMR08}.  For past-event
opacity this becomes:
\begin{quote}
for every \(w\in L(\M)\) and position \(i\), if
\(w,i\models\Diamond^-\mathit{sec}\), then there exists \(w'\in L(\M)\) with
\((w,i)\sim_a(w',i)\) and \(w',i\not\models\Diamond^-\mathit{sec}\).
\end{quote}

\begin{lemma}[Opacity as lack of knowledge]
\label{lem:opacity-equivalence}
System \(\M\) satisfies the trace condition above iff
\[
  \M\models \Box\neg K_a\Diamond^-\mathit{sec}.
\]
\end{lemma}

\begin{proof}
The proof is immediate: at a position with \(\Diamond^-\mathit{sec}\),
the formula \(\neg K_a\Diamond^-\mathit{sec}\) states exactly the existence of an
indistinguishable non-secret witness, and at a position without
\(\Diamond^-\mathit{sec}\) it holds trivially, witnessed by the trace itself.
\end{proof}

For timed automata with observable time stamps, current-location opacity is
undecidable~\cite{Cassez09}.

The formula $\Box\neg K_a\Diamond^-\mathit{sec}$ prevents the observer from
becoming certain that the secret occurred, but it does not prevent certainty
that it did not occur. Two-sided opacity is expressed by
\[
  \Box\neg\bigl(K_a\Diamond^-\mathit{sec}\vee
                 K_a\neg\Diamond^-\mathit{sec}\bigr).
\]
For a confidential binary value encoded at the beginning of a run, this
requires both values to remain compatible with the observations. Metric
variants can restrict the protection requirement to a bounded window.

In the next section, we prove \EXPSPACE-hardness even for formulas with a
single occurrence of a knowledge operator.

\section{EXPSPACE Lower Bound}
\label{sec:lower-bound}

We reduce from exponential-width corridor tiling. An instance asks for a tiling
of a grid of width $2^m$ and arbitrary finite height, with matching horizontal
and vertical colors and prescribed corner tiles. The model emits the grid in
row-major order, from the bottom row upward. Each cell is represented by a block
containing an $m$-bit column counter and a tile. After the last row, the trace enters a terminal
$\#^\omega$ loop. A single agent $a$ observes all propositions except two
auxiliary tags.

The counter, corner constraints, and horizontal compatibility are checked on
the visible trace by a polynomial-size formula. Vertical compatibility is the
only condition involving cells separated by an entire row. Their distance is
$2^m$ cells, so expanding this distance with next operators would require an
exponential-size formula, and a metric operator does not provide the direct
comparison either (cf.\ Remark~\ref{rem:metric-free}).

We check vertical compatibility with one knowledge operator. Two propositions
$\Lo$ and $\Up$ are hidden from $a$ and may each mark at most one cell start. A
placement with $\Lo$ on a cell and $\Up$ on the cell in the next row selects a
vertical edge. At the first $\#$-position, all placements over the same visible
prefix are indistinguishable to $a$. Thus $K_a$ quantifies over every candidate
placement. A bad vertical edge yields an indistinguishable placement for which
the test fails. Past operators compare the counter at the $\Lo$-cell with later
cells: the first later occurrence of the same counter is the cell in the next
row. Figure~\ref{fig:hidden-tag-column-check} shows the construction.

\begin{figure}[t]
\centering
\resizebox{.96\textwidth}{!}{%
\begin{tikzpicture}[
  x=1cm,y=1cm,
  font=\scriptsize,
  cell/.style={draw,minimum width=.82cm,minimum height=.55cm,align=center},
  tag/.style={circle,draw,fill=white,inner sep=1pt,font=\tiny},
  arr/.style={-{Stealth[length=2mm]},thick},
  note/.style={align=center}
]
  \node[cell]    (a0) at (0.50,0) {$0$};
  \node[cell]    (a1) at (1.40,0) {$1$};
  \node at (2.05,0) {$\cdots$};
  \node[cell] (ai) at (2.80,0) {$i$};
  \node at (3.45,0) {$\cdots$};
  \node[cell]    (am) at (4.40,0) {$2^m\!\!-\!\!1$};
  \node[cell]   (b0) at (5.50,0) {$0$};
  \node[cell]   (b1) at (6.40,0) {$1$};
  \node at (7.05,0) {$\cdots$};
  \node[cell] (bi) at (7.80,0) {$i$};
  \node at (8.45,0) {$\cdots$};
  \node[cell]   (bm) at (9.40,0) {$2^m\!\!-\!\!1$};
  \node[cell]   (e1) at (10.50,0) {$\#$};
  \node[cell]   (e2) at (11.40,0) {$\#$};
  \node[anchor=west] at (11.85,0) {$\cdots$};

  \node[tag] at (ai.north west) {$\Lo$};
  \node[tag] at (bi.north west) {$\Up$};

  \draw[arr] (ai.north) to[bend left=28] (bi.north);
  \node[note] at (5.30,1.55)
    {first later cell with the same counter $\;=\;$ the cell one row above};

  \draw[decorate,decoration={brace,amplitude=4pt,mirror}]
    (a0.south west) -- (am.south east)
    node[midway,below=4pt] {row $j$ \;(counter runs $0,\dots,2^m\!\!-\!\!1$)};
  \draw[decorate,decoration={brace,amplitude=4pt,mirror}]
    (b0.south west) -- (bm.south east)
    node[midway,below=4pt] {row $j{+}1$};

  \node[note,anchor=south] (endnote) at (10.50,0.78) {first $\#$:\\evaluate here};
  \draw[arr] (endnote.south) -- (e1.north);
\end{tikzpicture}%
}
\caption{Hidden-tag test on a row-major trace. The column counter cycles
through $0,\dots,2^m\!-\!1$ in each row, so the next cell with the same counter
is in the next row. The hidden tags $\Lo$ and $\Up$ select a candidate vertical
edge. At the first $\#$-position, knowledge quantifies over all tag placements,
where $\Pair$ identifies actual vertical edges and $\Good$ checks their colors.}
\label{fig:hidden-tag-column-check}
\end{figure}
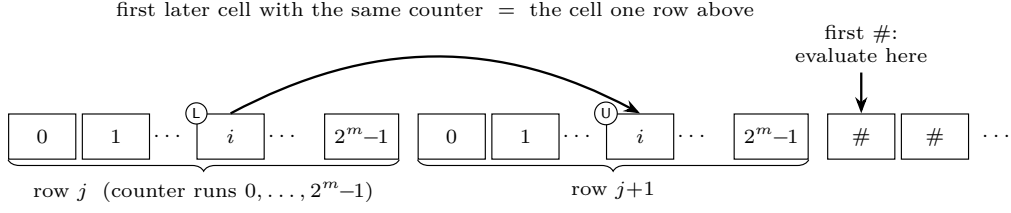

We state the reduction here and give the formulas and their semantic proofs in
Appendix~\ref{sec:lower-bound-details}. The source problem is the standard
exponential-width corridor-tiling problem (cf.~\cite{LMS2002}). An
instance
\[
  \mathcal{T}=(C,T,m,d_{\mathsf{init}},d_{\mathsf{final}})
\]
is given by a finite color set $C$, a set $T\subseteq C^4$ of tiles, each with a bottom, left, top, and right color, a width parameter $m\ge 1$ in unary, and two corner
tiles $d_{\mathsf{init}},d_{\mathsf{final}}\in T$.  It asks for a height $n\ge 1$
and a tiling of \(\{0,\ldots,2^m-1\}\times\{0,\ldots,n-1\}\) by tiles from $T$ so
that adjacent tiles agree on the color of their shared edge, the bottom-left
tile is $d_{\mathsf{init}}$, and the top-right tile is $d_{\mathsf{final}}$.
Since $m$ is in unary, the width $2^m$ is exponential in the size of $\mathcal{T}$,
and the problem is \EXPSPACE-complete.  The structure follows the tiling reduction of~\cite{LMS2002}. Their
forgettable-past operator is used to compare cells in adjacent rows, whereas our
reduction instead uses hidden tags quantified by knowledge.

For such an instance we build a polynomial-size B\"uchi automaton
\(\M_{\mathcal{T}}^{\mathsf{tag}}\) in which all states are accepting.  Its visible propositions
are the counter bits, the tile propositions, the terminal marker \(\#\), and a
marker \(\End\) that flags the first \(\#\)-position.  The only invisible propositions
are \(\Lo\) and \(\Up\).  The system ensures that the visible trace is a sequence of
complete cell blocks followed, when it stops, by \(\#^\omega\), that \(\End\) holds
exactly at the first \(\#\), and that each hidden tag occurs at most once and only
at a cell start. For each visible trace, every placement of \(\Lo\) and
\(\Up\) satisfying these restrictions is accepted.

The reduction formula is
\[
  \PhiI := \Base\wedge\Diamond(\End\wedge\Col),
  \qquad
  \Col := K_a(\Pair\rightarrow\Good).
\]
The knowledge operator ranges over words in
\(L(\M_{\mathcal{T}}^{\mathsf{tag}})\). Thus every epistemic alternative
already satisfies the uniqueness and cell-start restrictions on the hidden
tags, even though the agent cannot observe them.
The visible marker $\End$ occurs at the first $\#$-position, where $\Col$ is
evaluated. The formula $\Base$ checks the two corners, the modulo-$2^m$ column
counter, and horizontal compatibility. The formula $\Pair$ holds when
$\Lo$ marks a cell and $\Up$ marks the first later cell with the same counter.
Under $\Base$, these cells are vertical neighbors. The formula $\Good$
checks their tile colors. All formulas have size polynomial in $m+|T|$, and
their definitions are given in Appendix~\ref{sec:lower-bound-details}.

\begin{lemma}[Hidden-tag column test]
\label{lem:knowledge-column}
Let \(w\in L(\M_{\mathcal{T}}^{\mathsf{tag}})\) such that \(w,1\models\Base\),
and let \(e\) be the first \(\#\)-position.  Then
\[
  w,e\models K_a(\Pair\rightarrow\Good)
\]
iff all vertical neighbors in the visible grid prefix of \(w\) are compatible.
\end{lemma}

\begin{proof}[Proof sketch]
If some vertical edge is bad, place \(\Lo\) on its lower endpoint and \(\Up\) on
its upper endpoint.  This changes only invisible propositions, so the resulting
word is indistinguishable for \(a\) at \(e\).  In that word \(\Pair\) holds and
\(\Good\) fails, hence \(K_a(\Pair\rightarrow\Good)\) is false.  Conversely, if
all vertical edges are good, then every indistinguishable word differs only by
the hidden tag placement.  Whenever its tags satisfy \(\Pair\), they select a real
vertical edge, and that edge satisfies \(\Good\).
\end{proof}

\begin{theorem}[Reduction correctness]
\label{thm:reduction-correctness}
The tiling instance \(\mathcal{T}\) has a valid tiling of width \(2^m\) and some height
\(n\ge 1\) iff there exists \(w\in L(\M_{\mathcal{T}}^{\mathsf{tag}})\) such that
\[
  w,1\models_{\M_{\mathcal{T}}^{\mathsf{tag}}}\PhiI .
\]
\end{theorem}

\begin{proof}[Proof sketch]
If a valid tiling exists, encode it row by row, followed by \(\#^\omega\), and
choose any legal hidden tag placement.  The visible part satisfies \(\Base\).
All vertical adjacencies are compatible, so Lemma~\ref{lem:knowledge-column}
gives \(\Col=K_a(\Pair\rightarrow\Good)\) at the first \(\#\)-position.  Conversely, if some trace satisfies
\(\PhiI\), then \(\Base\) extracts a row-major grid with the right corners and
all horizontal constraints, while Lemma~\ref{lem:knowledge-column} provides all
vertical constraints.  Thus the visible prefix before the first \(\#\)-position
is a valid tiling.
\end{proof}

\begin{theorem}[\EXPSPACE-hardness]
\label{thm:expspace-hardness}
Model checking for $\KMTLone$ over finite B\"uchi automata under
synchronous perfect-recall knowledge is \EXPSPACE-hard.  The lower bound already
holds when all states of the automaton are accepting and the formula has one
agent, one occurrence of $K_a$, all past operators inside the scope of $K_a$,
and only the unbounded interval~$\Nat$.
\end{theorem}

\begin{proof}
The construction $\mathcal{T}\mapsto\M_{\mathcal{T}}^{\mathsf{tag}}$ is polynomial in the size of the
tiling instance.  The formula $\PhiI$ has size polynomial in $m+|T|$: all
iterations $X^k$ satisfy $k\le m+1$, the same-column test has $m$ conjuncts, and
$\Good$ has at most $|T|^2$ disjuncts.

By Theorem~\ref{thm:reduction-correctness}, $\mathcal{T}$ has a valid tiling iff there is
some trace satisfying $\PhiI$.  Therefore, for the model-checking
problem,
\[
  \mathcal{T}\text{ is positive}
  \quad\Longleftrightarrow\quad
  \M_{\mathcal{T}}^{\mathsf{tag}}\not\models \neg\PhiI.
\]
Since the source tiling problem is \EXPSPACE-complete and \EXPSPACE is closed
under complement, model checking is \EXPSPACE-hard.
All states of the constructed model are accepting. The formula $\neg\PhiI$
uses one agent and one occurrence of $K_a$, which contains all past operators.
The evaluation point is marked by the atomic proposition $\End$, and every
temporal interval is $\Nat$. This proves all the stated restrictions.
\end{proof}
\begin{remark}[Metric-freeness]
\label{rem:metric-free}
With standard metric until, a binary constant for the distance between two
vertically adjacent cell blocks would allow a direct comparison. Strict
first-time semantics does not provide such random access: a formula
$\Diamond_{[k,k]}\psi$ also requires the position at distance $k$ to be the
first future $\psi$-position. This is consistent with the \PSPACE upper bound for satisfiability of the knowledge-free
strict first-time fragment~\cite{Alsmann025}. Our reduction uses past and
knowledge instead of positive metric constants.
\end{remark}

\section{EXPSPACE Upper Bound}
\label{sec:upper-bound}

We construct an automaton that searches for a system execution violating the
input formula. The main difficulty is to evaluate knowledge along this
execution: the observer must account for all histories with the same
observations, including the temporal information needed to distinguish their
pasts. We first explain what the construction records and then give its temporal
and epistemic components.

For a knowledge formula $K_a\psi$ with knowledge-free $\psi$, consider the
temporal test automaton for $\neg\psi$ constructed below. After an observation prefix $u$, the
observer stores exactly the pairs $(q,s)$ reachable by a model run prefix and
a simultaneous test-automaton run prefix whose model letters have observation
$u$. Such a pair witnesses a failure of knowledge if the two run prefixes
have a common accepting continuation on which $\neg\psi$ holds at the current
position. Thus the construction separates reachability along the observed
prefix from the existence of an accepting counterexample continuation.

For nested knowledge, we replace inner knowledge subformulas by fresh
propositions. Their values are provided by the inner observers, both along
the prefix and along each candidate continuation. The decisive property is
that knowledge of the same agent has a common truth value on all histories
with the same observation prefix. We can therefore maintain one observer per
knowledge occurrence and update them from the innermost occurrences outward.
The following example explains why each observer must retain the temporal
component $s$ as well as the system state $q$.

\begin{example}[Why system states do not suffice]
\label{ex:past-needs-types}
Let $\AP=\{p,r\}$ and $\AP_a=\{r\}$. Consider a model with three initial branches
\[
\begin{gathered}
  q_0\xrightarrow{\{p\}}q_1\xrightarrow{\emptyset}q,
  \qquad
  q'_0\xrightarrow{\emptyset}q'_1\xrightarrow{\emptyset}q,\\
  q''_0\xrightarrow{\{r,p\}}q''_1\xrightarrow{\emptyset}q,
\end{gathered}
\]
followed by the loop $q\xrightarrow{\emptyset}q$. All states are accepting.
After two letters, the first two branches have observation prefix
$\emptyset\emptyset$, while the third has $\{r\}\emptyset$. Both observation
prefixes give the state information set $\{q\}$, but $K_a\Diamond^-p$ is false
for the first prefix and true for the second.

In the product observer, the information set after $\emptyset\emptyset$
contains pairs $(q,s_{\mathsf{seen}})$ and $(q,s_{\mathsf{unseen}})$, whose
temporal components record whether $p$ occurred. After $\{r\}\emptyset$, only
the first case is possible. The observers therefore use subsets of
$Q\times S_{\neg\Diamond^-p}$.
\end{example}

\paragraph*{Temporal test automata}
Alsmann and Lange~\cite{Alsmann025} prove \PSPACE-completeness of
satisfiability for the future-only strict first-time logic using lean
Hintikka sets and Streett acceptance. Here, we construct temporal test automata
that also handle past operators and test a formula at any position of a word.
For a knowledge-free formula $\theta$, the automaton $A_\theta$ reads from
the beginning of the word. A set $C_\theta$ distinguishes the states at
which $\theta$ is true: $w,i\models\theta$ iff some accepting run on $w$
visits $C_\theta$ after reading position $i$.

Let $B$ be the largest finite bound in $\theta$, or $0$ if there is none.
The automaton stores one counter for each metric subformula
$\gamma=\varphi U^1_I\psi$ or $\gamma=\varphi S^1_I\psi$, with values in
\[
  D_B=\{0,\ldots,B\}\cup\{{>}B,\bot\}.
\]
The counter records the distance to the first future, respectively most
recent past, $\psi$-position, provided that $\varphi\wedge\neg\psi$ holds
between the current position and the witness. The symbol ${>}B$ records a
distance exceeding $B$, and $\bot$ records the absence of such a witness.
The subformula holds iff its counter lies in $I$, where ${>}B$ belongs to
every interval $[\ell,\infty)$ occurring in $\theta$ and to no bounded
interval, and $\bot$ belongs to none. Thus two-sided intervals are handled
directly, and each counter uses $O(\log(B+2))$ bits.

Write $z\oplus1$ for the increment capped at ${>}B$, with
$({>}B)\oplus1={>}B$ and $\bot\oplus1=\bot$. The state also records truth
values for atomic propositions, next formulas, and yesterday formulas.
Boolean combinations are evaluated compositionally. Transitions enforce
the next/yesterday equivalences and the counter recurrences below, using
the truth values of the arguments in the corresponding states.

Next formulas and future counters are guessed. Transitions and acceptance
check these guesses, so on an accepting run every entry has its semantic
value. A finite run prefix may still carry a wrong guess, which is why
observers are combined with a test for an accepting continuation.

\paragraph*{Past operators}
For $\gamma=\varphi S^1_I\psi$, the counter at position $1$ is $0$ if
$\psi$ holds and $\bot$ otherwise. At later positions,
\[
  c_{i+1}(\gamma)=
  \begin{cases}
    0 & \text{if }\psi\text{ holds at }i+1,\\
    c_i(\gamma)\oplus1
      & \text{if }\varphi\wedge\neg\psi\text{ holds at }i+1,\\
    \bot & \text{otherwise}.
  \end{cases}
\]
Once its argument values are fixed, the counter is determined by this
forward recurrence. In particular, an old witness remains represented by
${>}B$ until $\psi$ resets the counter or $\neg\varphi\wedge\neg\psi$
removes it. Past counters need no acceptance condition.

\paragraph*{Future operators}
For $\gamma=\varphi U^1_I\psi$, transitions enforce the reverse recurrence
\[
  c_i(\gamma)=
  \begin{cases}
    0 & \text{if }\psi\text{ holds at }i,\\
    c_{i+1}(\gamma)\oplus1
      & \text{if }\varphi\wedge\neg\psi\text{ holds at }i,\\
    \bot & \text{otherwise}.
  \end{cases}
\]
This equation alone permits the value ${>}B$ forever without a witness.
For each future counter, a generalized B\"uchi condition therefore requires
infinitely many positions at which the counter is $\bot$ or $\psi$ holds.
Every active counter must then reach its witness. For example, on a suffix
where $\varphi\wedge\neg\psi$ holds forever, acceptance forces a later
$\bot$, which the recurrence propagates backwards to every earlier position
of that suffix.

There are linearly many counters and acceptance sets. The construction
therefore has exponentially many states with polynomial-space local tests.
Its correctness is proved in Appendix~\ref{sec:upper-bound-details}.

\begin{lemma}[Temporal test automata]
\label{lem:kfree-test-automata}
Let $\AP'$ be a finite set of propositions and let $\theta$ be a knowledge-free
KMTL formula over $\AP'$ with binary-encoded metric
intervals. There is a generalized B\"uchi automaton
\[
  A_\theta=(S_\theta,S_{0,\theta},\delta_\theta,\mathcal F_\theta)
\]
over $2^{\AP'}$, together with a set $C_\theta\subseteq S_\theta$, such that:
\begin{enumerate}[(i)]
\item $|S_\theta|\le 2^{\mathsf{poly}(|\theta|)}$ for a fixed polynomial
  $\mathsf{poly}$, and there are at most $|\theta|$ accepting sets. States,
  initial and transition tests, and membership in each accepting set and in
  $C_\theta$ are representable and testable in space polynomial in
  $|\theta|+|\AP'|$.
\item For every $w\in(2^{\AP'})^\omega$ and every $i\ge1$,
  \[
    w,i\models\theta
    \quad\Longleftrightarrow\quad
    \begin{array}{l}
    \text{some accepting run }s_0s_1s_2\ldots\text{ of }A_\theta
    \text{ over }w\\
    \text{satisfies }s_i\in C_\theta.
    \end{array}
  \]
\end{enumerate}
\end{lemma}

\paragraph*{Accepting continuations}
Fix $\M=(Q,Q_0,\Delta,F_\M)$ and a knowledge-free $\theta$ over $\AP$.
A pair $(q,s)$ is \emph{live for $\theta$}, written
$(q,s)\in\mathsf{Live}^\M_\theta$, if $s\in C_\theta$ and the product of
$\M$ and $A_\theta$ has a continuation from $(q,s)$ that is accepting in both
components. A product step advances both automata on the same letter.
Consequently, a model run prefix ending in $q$ extends to an accepted word
satisfying $\theta$ at the current position iff a test-automaton run on the
same prefix ends in some $s$ with $(q,s)\in\mathsf{Live}^\M_\theta$.
The product has exponentially many states of polynomial size, so this
continuation test uses polynomial space.

\paragraph*{Epistemic observers}
First consider $K_a\psi$ with knowledge-free $\psi$, and put $\chi=\neg\psi$.
For any observation prefix
$u\in(2^{\AP_a})^+$, define
\[
  O_{a,\psi}(u)\subseteq Q\times S_\chi
\]
as the set of pairs reachable by a run prefix of $\M$ with observation $u$ and
a simultaneous run prefix of $A_\chi$ on the same model letters.
On a new observation, we take every joint successor whose model letter has
that observation. Appendix~\ref{sec:upper-bound-details} gives the update and
its correctness proof.

For $w\in L(\M)$, $i\ge1$, and
$u=\obs_a(\pref{w}{i})$, the knowledge test is
\[
  w,i\models_\M K_a\psi
  \quad\Longleftrightarrow\quad
  O_{a,\psi}(u)\cap\mathsf{Live}^\M_\chi=\emptyset
\]
(Lemma~\ref{lem:knowledge-test}). A pair in the intersection represents an
indistinguishable prefix with an accepting continuation on which $\neg\psi$
holds at position $i$.

\paragraph*{Nested knowledge of the same agent}
Our fragment $\KMTLone$ excludes nesting between different agents, but permits repeated
nesting for the same agent. For fixed $a$, the truth of $K_a\xi$ is invariant on
an $a$-observation class:
\[
  (w,i)\sim_a(w',i)
  \quad\Longrightarrow\quad
  \bigl(w,i\models_\M K_a\xi\iff w',i\models_\M K_a\xi\bigr).
\]
We therefore process nested knowledge from the inside out. For an occurrence
$\eta=K_a\psi$, obtain $\widehat\psi$ by replacing each maximal knowledge
subformula of $\psi$ with a fresh proposition. Their
valuations depend only on the current observation prefix, and
$\widehat\psi$ is knowledge-free over the extended proposition set.
The formula $\eta$ is false iff $\psi$ is false at the current position on
some indistinguishable accepted word. We use the temporal test automaton $A_\eta$ for
$\chi_\eta:=\neg\widehat\psi$, which reads model letters with the fresh
valuations provided by the inner observers. Write $S_\eta$, $C_\eta$, and
$O_\eta$ for its state set,
distinguished set, and observer, respectively.

\paragraph*{Relative liveness}
The preceding observer handles the prefix up to the current position. If
$\widehat\psi$ contains an inner knowledge proposition below a future
operator, its value is also needed along a possible counterexample
continuation. That value depends on the continuation of the observation
history, not only on the pair $(q,s)$.

For $\eta=K_a\psi$, the test asks whether a given pair $(q,s)$ has an
accepting continuation on which $\psi$ is false at the current position.
We use a graph with configurations
$(q,s,\mathcal I)$, where $\mathcal I$ contains the observers for every
knowledge subformula strictly inside $\eta$. All nesting levels are retained
because an observer can be updated only after the observers inside its own body
have provided their current truth values. A transition chooses a model letter
$\tau$, advances $q$, and updates $\mathcal I$ on $\obs_a(\tau)$ from the
innermost observer outward. The updated tuple $\mathcal I'$ determines the
valuation $\lambda_\eta(\mathcal I')$ read by $A_\eta$: the proposition for
an immediate inner occurrence $\zeta$ is true iff its observer in
$\mathcal I'$ contains no pair that is live relative to the observers
inside $\zeta$. Relative liveness is thus defined by induction on nesting
depth. We write
\[
  \mathsf{Live}^\M_\eta(q,s,\mathcal I)
\]
when $s\in C_\eta$ and the resulting graph has a continuation from
$(q,s,\mathcal I)$ that is accepting in both $\M$ and $A_\eta$. At every future position, the inner
valuation is shared by all runs with the same extended observation
prefix. If $\psi$ contains no knowledge subformula, $\mathcal I$ is empty and
the definition reduces to the ordinary product liveness test. Formal
definitions and correctness are given in
Appendix~\ref{sec:upper-bound-details}.

\begin{example}[Two inner observers]
\label{ex:two-observers}
Consider
\[
  \Box\,
  \underbrace{K_a\bigl(\Diamond\,
    \underbrace{K_a(r \wedge \Diamond\,
      \underbrace{K_a p}_{\eta_2})}_{\eta_1}\bigr)}_{\eta_0}.
\]
Introduce fresh propositions $x_j$ for the occurrences $\eta_j$, from the
innermost occurrence outward:
\[
  x_2 := K_a p,
  \qquad
  x_1 := K_a(r \wedge \Diamond\, x_2),
  \qquad
  x_0 := K_a(\Diamond\, x_1),
\]
so the negation of the displayed formula has the knowledge-free abstraction
$\theta_{\mathsf{top}}=\neg\Box\,x_0$. To evaluate $x_0$, the automaton
$A_{\eta_0}$ tests $\chi_{\eta_0}=\neg\Diamond x_1\equiv\Box\neg x_1$.
As it moves forward, this
automaton needs the value of $x_1$ at every position, while
$x_1=K_a(r\wedge\Diamond x_2)$ in turn depends on $x_2$. Hence the relative
liveness state for the outer occurrence carries both inner observers:
\[
  \mathcal I_{\eta_0}(u)=
  \bigl(O_{\eta_1}(u),\;O_{\eta_2}(u)\bigr).
\]
The observer $O_{\eta_2}(u)$ determines $x_2$. To determine $x_1$, we test
the pairs in $O_{\eta_1}(u)$ for liveness relative to $O_{\eta_2}(u)$.
A configuration of the relative-liveness graph for $\eta_0$ is
$(q,s,\mathcal I)$, where $s\in S_{\eta_0}$. One future step on a
letter $\tau$ proceeds from the innermost occurrence outward:
\begin{enumerate}[(1)]
\item advance $O_{\eta_2}$ on $\obs_a(\tau)$ and set $x_2$ true iff the
  updated observer contains no live pair.
\item advance $O_{\eta_1}$ on $\obs_a(\tau)$, using the new value of $x_2$,
  and determine $x_1$ by testing liveness relative to the updated $O_{\eta_2}$.
\item feed $A_{\eta_0}$ the letter enriched with the new value of $x_1$.
\end{enumerate}
Retaining both inner observers makes these updates possible along each
candidate continuation.
\end{example}

\paragraph*{Global search}
We use the same product construction to search for a word in $L(\M)$
satisfying $\neg\varphi$, with relative-liveness tests as subroutines.
Replace every maximal knowledge subformula $\eta$ of $\neg\varphi$ by a fresh
proposition $p_\eta$, obtaining the knowledge-free formula
$\theta_{\mathsf{top}}$. Build its temporal test automaton over $\AP$
extended by the fresh propositions for all knowledge occurrences.
A global configuration is
\[
  (q,t,\mathcal O),
\]
where $q\in Q$, $t\in S_{\theta_{\mathsf{top}}}$, and $\mathcal O$ contains
the observers for all knowledge occurrences, including nested descendants.
These observers may belong to different agents. The knowledge tests above compute from
$\mathcal O$ a valuation $\lambda_{\mathcal O}$ of the fresh propositions.
These tests may follow continuations different from the word being
constructed by the global search.

At each step, choose a model transition, update each observer on the observation
of its agent from the innermost occurrences outward, and advance the top
automaton on the model letter enriched with the newly computed knowledge values. Initially, both
automata start in initial states, and the first state reached in the top
automaton must belong to $C_{\theta_{\mathsf{top}}}$, since the counterexample
formula is tested at position $1$. The acceptance condition requires both
the model and the top automaton to be accepting. Acceptance of the inner
temporal automata is checked inside the relative liveness predicates.

The next lemma identifies the observer contents and computed knowledge values.

\begin{lemma}[Global run invariant]
\label{lem:global-invariant}
Let an accepting run of the global product follow
$w=\sigma_1\sigma_2\ldots\in L(\M)$, and write
$(q_i,t_i,\mathcal O_i)$ for its configuration after reading
$\pref{w}{i}$. For every $i\ge1$:
\begin{enumerate}[(i)]
\item For every knowledge occurrence $\eta=K_a\psi$, the corresponding component
  of $\mathcal O_i$ is $O_\eta\bigl(\obs_a(\pref{w}{i})\bigr)$.
\item For every knowledge occurrence $\eta$, we have
  $p_\eta\in\lambda_{\mathcal O_i}$ iff $w,i\models_\M\eta$.
\end{enumerate}
\end{lemma}

\begin{proof}
We prove (i) and (ii) together by induction on the knowledge-nesting depth,
for all positions at once. For an innermost occurrence, the ordinary subset
update gives (i), and Lemma~\ref{lem:knowledge-test} gives (ii).

For $\eta$ with inner knowledge, the induction hypothesis supplies (i) and
(ii) for every proper descendant. The update of $O_\eta$ therefore reads
the correct inner valuations. Induction on the prefix length and the observer
recursion give (i) for $\eta$. Item~(ii) is then
Lemma~\ref{lem:knowledge-test}, which applies to every observation prefix.
\end{proof}

By item~(ii), at every position the enriched word read by the top automaton
carries the correct truth values of the knowledge subformulas that were replaced. An
induction on the structure of $\neg\varphi$ therefore shows that the enriched
word satisfies $\theta_{\mathsf{top}}$ at position~$1$ iff
$w,1\models_\M\neg\varphi$. The top component of an accepting global run
reads this enriched word and visits $C_{\theta_{\mathsf{top}}}$ after its
first letter. Lemma~\ref{lem:kfree-test-automata} therefore yields
$w,1\models_\M\neg\varphi$. Conversely, fix such a word
$w\in L(\M)$ and an accepting model run. Initialize and update every observer
along $w$. The same joint induction gives the semantic knowledge values.
The enriched word therefore satisfies $\theta_{\mathsf{top}}$, and
Lemma~\ref{lem:kfree-test-automata} supplies an accepting top-automaton run
whose state after the first letter belongs to $C_{\theta_{\mathsf{top}}}$.
These components form an accepting global run.

\begin{theorem}[\EXPSPACE upper bound]
\label{thm:expspace-upper}
Model checking for $\KMTLone$ over
B\"uchi automata, with past and binary-encoded metric intervals, is in \EXPSPACE.
\end{theorem}

\begin{proof}
Let $n$ be the total input length, including the observation sets. Propositions,
agents, states, and transitions are listed explicitly, with letters represented
by bit vectors. We construct a temporal test automaton
for $\theta_{\mathsf{top}}$ and for the abstracted counterexample body of every
knowledge subformula. There are linearly many such formulas. By
Lemma~\ref{lem:kfree-test-automata}, every automaton has at most
$2^{\mathsf{poly}(n)}$ states, and one state and all local predicates use
polynomial space.

An observer is a subset of $Q\times S_\theta$ and can therefore be represented
by a bit vector of length at most $2^{\mathsf{poly}(n)}$. Since the number of
knowledge subformulas is linear, a global configuration uses exponential
space.

Successors of a global configuration are generated in exponential space. The
procedure chooses the next model transition, updates the observer bit vectors,
and computes the fresh knowledge propositions. To evaluate one such
proposition, it scans the corresponding observer and tests relative liveness
for each candidate pair. A relative-liveness configuration
$(q,s,\mathcal I)$ also has an exponential-space representation. Its successor
procedure updates the inner observers and may call only liveness tests for
strictly smaller knowledge subformulas. Each relative-liveness graph has at
most doubly exponentially many vertices. Its generalized B\"uchi nonemptiness
test is decided deterministically in exponential space by Savitch's
theorem~\cite{savitch1970}. This permits its result to be negated when
computing a knowledge value. More precisely, at each nesting level the
reachability procedure stores an exponential amount of data and calls the
next level only to evaluate an edge. The space for the suspended calls adds
over the linear nesting depth, so the total remains exponential.

The implicit global graph consequently has at most
$2^{2^{\mathsf{poly}(n)}}$ configurations and an exponential-space successor
procedure. Nonemptiness is checked by guessing a lasso whose loop visits the model's
accepting set and every accepting set of the top automaton. Only a constant
number of configurations, the polynomially many acceptance flags, and a
path-length counter of exponentially many bits are stored. This yields a
nondeterministic exponential-space procedure, which lies in deterministic
\EXPSPACE by Savitch's theorem~\cite{savitch1970}. By
Lemma~\ref{lem:global-invariant} and the argument following it, the global product is nonempty
exactly when some $w\in L(\M)$ satisfies $\neg\varphi$ at position~$1$, so
$\M\models\varphi$ holds iff the product is empty. Since \EXPSPACE is closed
under complement, model checking is in \EXPSPACE.
\end{proof}

\begin{corollary}[Main complexity theorem]
\label{cor:main-complexity}
Model checking for $\KMTLone$ over
B\"uchi automata, with past and binary-encoded metric intervals, is
\EXPSPACE-complete.
\end{corollary}

\begin{proof}
The lower bound is Theorem~\ref{thm:expspace-hardness}. The upper bound is
Theorem~\ref{thm:expspace-upper}.
\end{proof}

Table~\ref{tab:complexity-landscape} situates this result among related logics.
For the future knowledge-free metric fragment, satisfiability is
\PSPACE-complete~\cite{Alsmann025}, and the model-checking result follows
using the standard product construction.

\begin{table}[t]
\centering
\renewcommand{\arraystretch}{1.25}
\caption{Model checking for $\KMTLone$ and neighboring fragments. The last row
is the result of this paper.}
\label{tab:complexity-landscape}
\begin{tabular}{@{}p{0.5\textwidth}p{0.18\textwidth}l@{}}
\toprule
\textbf{Logic} & \textbf{Problem} & \textbf{Complexity}\\
\midrule
Epistemic LTL, future, perfect recall, unrestricted alternation~\cite{BozzelliMM24}
  & model checking & nonelementary\\
Epistemic LTL, future, perfect recall, alternation depth one~\cite{BozzelliMM24}
  & model checking & \PSPACE-complete\\
Metric temporal logic, future, strict first-time, no knowledge~\cite{Alsmann025}
  & satisfiability, model checking & \PSPACE-complete\\
LTL with past and a forgettable-past operator~\cite{LMS2002}
  & satisfiability, model checking & \EXPSPACE-complete\\
\midrule
$\KMTLone$: past, metric intervals, perfect recall, alternation depth one
  & model checking & \EXPSPACE-complete\\
\bottomrule
\end{tabular}
\end{table}

\section{Discussion and Conclusion}
\label{sec:conclusion}

We have shown that model checking for agent-alternation-free epistemic metric
temporal logic with past is \EXPSPACE-complete under synchronous perfect recall.

The same-agent restriction enters the upper bound when inner knowledge
formulas are replaced by propositions shared by an observation class. With
an inner operator $K_b$ for $b\ne a$, histories indistinguishable to $a$ may
have different $b$-observation histories, so a single shared inner valuation
no longer suffices. In the future-only metric-free setting, fixed alternation
depth yields an elementary complexity hierarchy, while unrestricted
alternation gives nonelementary model checking~\cite{BozzelliMM24}. One
additional alternation level therefore stays elementary in that setting.
Determining the corresponding hierarchy
with past and strict first-time metric operators remains a question for
further work.

For future-only metric temporal logic over discrete time, with standard until
semantics and binary-encoded bounds allowing exact distances, model checking is
\EXPSPACE-complete~\cite{LaroussinieMS06}. For its extension with
past and agent-alternation-free perfect-recall knowledge, we conjecture
$2$-\EXPSPACE-completeness.

Monitorability asks whether every reachable observation prefix has a finite
continuation forcing a definitive verdict~\cite{PnueliZ06}.
Over dense time it is undecidable for nondeterministic timed-automaton
specifications~\cite{GrosenKL025}. For B\"uchi-automaton specifications it is
\PSPACE-complete~\cite{DiekertMW15}. Its complexity for $\KMTLone$
specifications under partial observation remains open.

\clearpage
\bibliography{lit}

\clearpage

\appendix

\section{Lower-Bound Details}
\label{sec:lower-bound-details}

We give the formulas used in \Cref{sec:lower-bound} and prove the properties
claimed there. Fix a tiling instance
\[
  \mathcal T=(C,T,m,d_{\mathsf{init}},d_{\mathsf{final}}),
\]
where \(T\subseteq C^4\) and \(m\ge 1\). For \(d\in T\), write
\(d_{\mathsf{down}},d_{\mathsf{left}},d_{\mathsf{up}},d_{\mathsf{right}}\)
for its four colors.

\subsection{Visible Grid Encodings}
\label{sec:grid-encodings}

Let
\[
  \Main=\{b_i^+,b_i^-\mid 1\le i\le m\}\cup T\cup\{\#,\End\}
\]
be a set of propositions. The visible proposition \(\End\) marks the first
\(\#\)-position.
A cell is represented by a block of length \(m+1\):
\[
  \{b_1^{\pm}\}\;\{b_2^{\pm}\}\;\ldots\;\{b_m^{\pm}\}\;\{d\}.
\]
At position \(i\) of the block, exactly one of \(b_i^+\) and \(b_i^-\) is
visible, and the last position carries exactly one tile proposition \(d\in T\).
The least significant bit is \(b_1\). Encoding one bit per position keeps the
number of letters enabled at each state polynomial in \(m+|T|\).

The visible part of a finite grid is listed row by row, from left to right and
from the bottom row upwards. It is followed by \(\#^\omega\). At a cell-start
position, put
\[
  \Cell := b_1^+\vee b_1^-.
\]
For \(1\le r\le m\) and \(d\in T\), define
\[
  \Bit_r^+ := X^{r-1}b_r^+,
  \qquad
  \Bit_r^- := X^{r-1}b_r^-,
  \qquad
  \Tile_d := X^m d.
\]
Evaluated at a cell start, \(\Bit_r^\pm\) and \(\Tile_d\) thus refer to the
\(r\)-th counter bit and the tile of that cell. Finally, let
\[
  \Zero := \bigwedge_{r=1}^m \Bit_r^-,
  \qquad
  \Max := \bigwedge_{r=1}^m \Bit_r^+.
\]

\subsection{The Tagged System}
\label{sec:tagged-system}

We now extend the set of propositions with the two hidden tags:
\[
  \AP_{\mathcal T}=\Main\cup\{\Lo,\Up\}.
\]
There is one agent \(a\), with
\[
  \AP_a=\Main.
\]
Hence \(a\) observes the complete grid encoding but not the tags.

\begin{definition}[Tagged automaton]
The tagged model \(\M_{\mathcal T}^{\mathsf{tag}}\) is a B\"uchi automaton over
\(2^{\AP_{\mathcal T}}\) with all states accepting, specified by the regular
language it recognizes: the infinite words satisfying the following conditions:
\begin{enumerate}[(a)]
\item The word either consists of an infinite sequence of complete cell
      blocks, or of a finite sequence of complete cell blocks followed by the
      terminal segment described in (b).
\item Within a cell block, the visible letters have the form specified above.
      At the first position of the terminal loop, the visible letter is
      \(\{\#,\End\}\), and at every later position it is \(\{\#\}\).
\item Neither tag occurs at a \(\#\)-position.
\item Each of \(\Lo\) and \(\Up\) occurs at most once, and only at a cell-start
      position.
\end{enumerate}
\end{definition}

To recognize this language, \(\M_{\mathcal T}^{\mathsf{tag}}\) needs only to
record the position within a cell block, whether the terminal loop has started,
and whether each tag has already occurred, so it has polynomially many states.
Its transition relation is polynomial as well.
Although the alphabet \(2^{\AP_{\mathcal T}}\) is exponential, the encoding of
one bit per position enables at each state only the \(O(|T|)\) letters that
carry a single counter bit or a single tile, together with one of the
constantly many admissible tag placements. Every other letter has no outgoing
transition and is rejected, so \(\delta\) is never enumerated over the full
alphabet. Hence \(\M_{\mathcal T}^{\mathsf{tag}}\) has size polynomial in
\(m+|T|\). Traces containing no \(\#\)-position are allowed, but cannot satisfy
the eventual \(\End\)-requirement in \(\PhiI\).

\begin{lemma}[Tag completeness]
\label{lem:tag-completeness}
Let \(w\in L(\M_{\mathcal T}^{\mathsf{tag}})\), let \(e\) be its first
\(\#\)-position, and let \(c,c'<e\) be cell-start positions. There is a word
\(w'\in L(\M_{\mathcal T}^{\mathsf{tag}})\) with the same visible projection
as \(w\) such that \(\Lo\) occurs exactly at \(c\), \(\Up\) occurs exactly at
\(c'\), and
\[
  (w,e)\sim_a(w',e).
\]
\end{lemma}

\begin{proof}
Keep the projection to \(\Main\) unchanged and place \(\Lo\) at \(c\) and \(\Up\)
at \(c'\). Since the language consists of all words satisfying conditions
(a)--(d), and this placement respects (c) and (d), the resulting \(w'\) is
accepted. As neither tag is visible to \(a\), the observation prefixes of \(w\)
and \(w'\) up to position \(e\) coincide.
\end{proof}

\subsection{The Base Formula}
\label{sec:base-formula}

The base formula checks the visible constraints that do not compare adjacent
rows. For \(1\le r\le m\), let
\[
  \Carry_r := \bigwedge_{s<r}\Bit_s^+,
\]
where \(\Carry_1=\top\), and put
\[
  \NextCell(\varphi):=X^{m+1}\varphi.
\]
Define
\[
\Inc :=
\bigwedge_{r=1}^m
\left(
\bigl(\Carry_r\rightarrow
      (\NextCell(\Bit_r^+)\leftrightarrow \Bit_r^-)\bigr)
\wedge
\bigl(\neg\Carry_r\rightarrow
      (\NextCell(\Bit_r^+)\leftrightarrow \Bit_r^+)\bigr)
\right).
\]
At a cell start, \(\Inc\) states that the counter in the next cell is the
current counter plus one modulo \(2^m\).

Horizontal compatibility is expressed by
\[
\mathsf{Horiz}:=
\Box\left(
\Cell\wedge\neg\Max
\rightarrow
\bigwedge_{d\in T}
\left(
\Tile_d\rightarrow
\bigvee_{\substack{d'\in T\\ d_{\mathsf{right}}=d'_{\mathsf{left}}}}
\NextCell(\Tile_{d'})
\right)
\right).
\]
Now let
\[
\Base :=
(\Cell\wedge\Zero\wedge\Tile_{d_{\mathsf{init}}})
\wedge
\Diamond(\Cell\wedge\Max\wedge\Tile_{d_{\mathsf{final}}}
                 \wedge\NextCell(\End))
\]
\[
\wedge\;
\Box(\Cell\wedge\NextCell(\neg\#)\rightarrow\Inc)
\wedge
\mathsf{Horiz}.
\]

\begin{lemma}[Meaning of the base formula]
\label{lem:base-meaning}
Let \(w\in L(\M_{\mathcal T}^{\mathsf{tag}})\). Then
\(w,1\models_{\M_{\mathcal T}^{\mathsf{tag}}}\Base\) iff \(w\) has a first
\(\#\)-position and the visible prefix before it is a nonempty sequence of complete rows of
width \(2^m\), listed in row-major order, with counter \(x\) in column \(x\)
for every \(x\in\{0,\ldots,2^m-1\}\), such that:
\begin{enumerate}[(i)]
\item its first cell has counter \(0\) and tile \(d_{\mathsf{init}}\),
\item its last cell has counter \(2^m-1\) and tile \(d_{\mathsf{final}}\), and
\item horizontally adjacent tiles have matching colors.
\end{enumerate}
\end{lemma}

\begin{proof}
The first conjunct of \(\Base\) fixes the first counter and tile. Before the
terminal loop, the increment conjunct forces each successive cell counter to
increase by one modulo \(2^m\). The conjunct containing \(\End\) requires the
terminal loop to begin immediately after a cell with counter \(2^m-1\) and tile
\(d_{\mathsf{final}}\). The number of cells before the loop is therefore a
positive multiple of \(2^m\), and these cells form complete rows. Finally,
\(\mathsf{Horiz}\) checks precisely the edges between consecutive columns, but it
does not compare the last cell of one row with the first cell of the next.
The converse follows by evaluating the four conjuncts on any prefix with the
stated properties.
\end{proof}

\subsection{The Hidden-Tag Column Check}
\label{sec:column-check}

The vertical test is evaluated at the first \(\#\)-position. At a cell start,
the following formula compares the current counter with that of the
\(\Lo\)-marked cell:
\[
\SameL :=
\Cell\wedge \Diamond^{-}(\Lo\wedge\Cell)
\wedge
\bigwedge_{r=1}^m
\left(
\Bit_r^+
\leftrightarrow
\Diamond^{-}(\Lo\wedge\Cell\wedge\Bit_r^+)
\right).
\]
Define
\[
\Pair :=
\Diamond^{-}
\left(
\Lo\wedge\Cell
\wedge
\NextCell
\left(
(\neg \#\wedge\neg\SameL)
\,U^1_{\Nat}\,
(\Up\wedge\Cell\wedge\SameL)
\right)
\right).
\]
In the left-hand argument, \(\neg\SameL\) fails at any earlier cell that shares
the counter of the \(\Lo\)-cell. Consequently, \(\Pair\) holds exactly when
\(\Up\) marks the first later cell whose counter agrees with that of the
\(\Lo\)-cell.

Vertical color compatibility is expressed by
\[
\Good :=
\Diamond^{-}
\left(
\Up\wedge\Cell
\wedge
\bigvee_{\substack{d,d'\in T\\ d_{\mathsf{up}}=d'_{\mathsf{down}}}}
\left(
\Tile_{d'}
\wedge
\Diamond^{-}(\Lo\wedge\Cell\wedge\Tile_d)
\right)
\right).
\]
Finally, set
\[
  \Col := K_a(\Pair\rightarrow\Good),
  \qquad
  \PhiI := \Base\wedge\Diamond(\End\wedge\Col).
\]

\begin{lemma}[Same-column predicate]
\label{lem:same-column}
Suppose \(w\in L(\M_{\mathcal T}^{\mathsf{tag}})\) has a unique
\(\Lo\)-marked cell at position \(c\). For every cell-start position \(x\ge c\),
\[
  w,x\models_{\M_{\mathcal T}^{\mathsf{tag}}}\SameL
\]
iff the cells at \(c\) and \(x\) carry the same counter.
\end{lemma}

\begin{proof}
For each \(r\), the formula
\(\Diamond^{-}(\Lo\wedge\Cell\wedge\Bit_r^+)\) holds at \(x\) iff bit \(r\)
of the unique \(\Lo\)-cell is \(1\). The \(r\)-th equivalence in \(\SameL\)
therefore states that the two cells agree on bit \(r\). The initial
\(\Cell\)-conjunct ensures that the formula is true only at cell starts.
\end{proof}

\begin{lemma}[Pair selects vertical neighbors]
\label{lem:pair-selects}
Assume \(w,1\models\Base\), let \(e\) be the first \(\#\)-position, and suppose
that \(\Lo\) and \(\Up\) occur uniquely at cell starts \(c\) and \(c'\),
respectively. Then
\[
  w,e\models\Pair
\]
iff \(c'\) is the first cell start after \(c\) whose counter equals the counter
at \(c\). Whenever such a cell exists, it is the cell directly above \(c\) in
the next row.
\end{lemma}

\begin{proof}
Since \(\Lo\) occurs only at \(c\), the outer \(\Diamond^{-}\) can refer only to
\(c\), and \(\NextCell\) moves to the start of the following cell. By
Lemma~\ref{lem:same-column}, \(\SameL\) holds at a later cell start exactly when
its counter equals that of the \(\Lo\)-cell. Before the first such cell, the
left-hand argument \(\neg\#\wedge\neg\SameL\) holds, and there it fails, so the
strict first-time until can hold only if this cell also carries \(\Up\). The
conjunct \(\Cell\) in the right-hand argument excludes a witness in the
terminal loop. Under \(\Base\), equal
counters recur once per row, so the first recurrence lies in the next row.
\end{proof}

\begin{lemma}[Good checks vertical compatibility]
\label{lem:good-checks}
Assume that \(w\) has unique tags and that \(w,e\models\Pair\) at its first
\(\#\)-position. Then
\[
  w,e\models\Good
\]
iff the tile at the \(\Lo\)-cell is vertically compatible with the tile at the
\(\Up\)-cell.
\end{lemma}

\begin{proof}
Since the tags are unique, \(w,e\models\Good\) iff the \(\Up\)-cell carries a
tile \(d'\) and the \(\Lo\)-cell carries a tile \(d\) with
\(d_{\mathsf{up}}=d'_{\mathsf{down}}\). This is exactly vertical compatibility
of the two tiles.
\end{proof}

\begin{proof}[Proof of Lemma~\ref{lem:knowledge-column}]
Suppose that a vertical edge is incompatible, with lower endpoint \(c\) and
upper endpoint \(c'\). By Lemma~\ref{lem:tag-completeness}, there is an
\(a\)-indistinguishable word \(w'\) with \(\Lo\) at \(c\) and \(\Up\) at
\(c'\). Lemma~\ref{lem:pair-selects} gives \(w',e\models\Pair\), whereas
Lemma~\ref{lem:good-checks} gives \(w',e\not\models\Good\). Hence
\(w,e\not\models K_a(\Pair\rightarrow\Good)\).

Conversely, assume that all vertical edges are compatible and let \(w'\in L(\M_{\mathcal T}^{\mathsf{tag}})\) satisfy
\((w,e)\sim_a(w',e)\). Since every grid proposition and \(\#\) is visible,
\(w'\) has the same visible prefix as \(w\) up to \(e\), and only the tag placement
may differ. If \(w',e\not\models\Pair\), the implication is true. If
\(w',e\models\Pair\), Lemma~\ref{lem:pair-selects} identifies a vertical edge
of the common grid, and Lemma~\ref{lem:good-checks} yields
\(w',e\models\Good\). Thus every indistinguishable word satisfies the
implication.
\end{proof}

\section{Upper-Bound Details}
\label{sec:upper-bound-details}

We now give the temporal test automata, the observer update, and the global
product used in the upper bound.

\subsection{Knowledge-Free Temporal Automata}

We give the counter construction used in the main text. Its states record
truth values and capped witness distances for the subformulas of $\theta$.
We retain generalized B\"uchi acceptance throughout.

\begin{definition}[Closure and counter domain]
\label{def:upper-closure}
Let $\mathsf{FL}(\theta)$ be the set of subformulas of $\theta$, and let
$B_\theta$ be its largest finite metric bound, or $0$ if none occurs.
Put
\[
  D_\theta=\{0,\ldots,B_\theta\}\cup\{{>}B_\theta,\bot\}.
\]
For $z\in D_\theta$, define
\[
 z\oplus1=
 \begin{cases}
   z+1 & \text{if }z\in\Nat\text{ and }z<B_\theta,\\
   {>}B_\theta & \text{if }z=B_\theta\text{ or }z={>}B_\theta,\\
   \bot & \text{if }z=\bot.
 \end{cases}
\]
For each interval $I$ in $\theta$, the relation $z\in I$ has its ordinary
meaning for numeric values. The symbol ${>}B_\theta$ belongs precisely to
the unbounded intervals, and $\bot$ belongs to no interval.
\end{definition}

\begin{definition}[Explicit-counter temporal types]
\label{def:upper-types}
A temporal type is a pair $H=(v,c)$. The map $v$ assigns a Boolean value
to every atomic proposition, next formula, and yesterday formula in
$\mathsf{FL}(\theta)$. The map $c$ assigns a value in $D_\theta$ to every
metric subformula. Define $H\Vdash\delta$ compositionally for Boolean
formulas, by $v(\delta)=1$ for atomic, next, and yesterday formulas, and by
\[
  H\Vdash\gamma\quad\Longleftrightarrow\quad c(\gamma)\in I
\]
when $\gamma=\varphi U^1_I\psi$ or $\gamma=\varphi S^1_I\psi$.
Write $\mathsf H_\theta$ for the set of temporal types.
\end{definition}

There are at most $|\theta|$ Boolean entries and counters. Since
$\log(B_\theta+2)=O(|\theta|)$, a type has polynomial bit length and
$|\mathsf H_\theta|\le 2^{\mathsf{poly}(|\theta|)}$.

For arguments $\varphi,\psi$, a future witness distance at $(w,i)$ is the
unique $d\ge0$, if it exists, such that $w,i+d\models\psi$ and
$w,k\models\varphi\wedge\neg\psi$ for $i\le k<i+d$.
A past witness distance is the unique $0\le d<i$ such that
$w,i-d\models\psi$ and $w,k\models\varphi\wedge\neg\psi$ for $i-d<k\le i$.
The truthful counter is $d$ when $d\le B_\theta$, ${>}B_\theta$ when
$d>B_\theta$, and $\bot$ when no witness exists. Interval membership
preserves the truth value of each metric subformula.

\begin{definition}[Initial and successor types]
\label{def:upper-type-step}
A type $H=(v,c)$ is prefix-initial if $H\nVdash Y\varphi$ for every
yesterday subformula and, for every $\gamma=\varphi S^1_I\psi$,
\[
 c(\gamma)=
 \begin{cases}
  0 & \text{if }H\Vdash\psi,\\
  \bot & \text{otherwise}.
 \end{cases}
\]
For $H=(v,c)$, $H'=(v',c')$, and a letter $\sigma$, write
$H\xrightarrow{\sigma}_\theta H'$ if the following conditions hold:
\begin{enumerate}[(i)]
\item $v'(p)=1$ iff $p\in\sigma$ for every atomic subformula $p$.
\item $H\Vdash X\varphi$ iff $H'\Vdash\varphi$ for every next subformula.
\item $H'\Vdash Y\varphi$ iff $H\Vdash\varphi$ for every yesterday subformula.
\item For every $\gamma=\varphi U^1_I\psi$,
\[
 c(\gamma)=
 \begin{cases}
  0 & \text{if }H\Vdash\psi,\\
  c'(\gamma)\oplus1 & \text{if }H\Vdash\varphi\wedge\neg\psi,\\
  \bot & \text{otherwise}.
 \end{cases}
\]
\item For every $\gamma=\varphi S^1_I\psi$,
\[
 c'(\gamma)=
 \begin{cases}
  0 & \text{if }H'\Vdash\psi,\\
  c(\gamma)\oplus1 & \text{if }H'\Vdash\varphi\wedge\neg\psi,\\
  \bot & \text{otherwise}.
 \end{cases}
\]
\end{enumerate}
\end{definition}

The label $\sigma$ describes the next position. The forward recurrence
initializes and propagates past counters, while the reverse recurrence
relates a future counter to its value one step closer to the witness.

\begin{definition}[Temporal test automaton]
\label{def:upper-test-automaton}
Add a pre-initial state $\iota_\theta$ to $\mathsf H_\theta$ and put
\[
 A_\theta=(S_\theta,\{\iota_\theta\},\delta_\theta,\mathcal F_\theta),
 \qquad S_\theta=\{\iota_\theta\}\cup\mathsf H_\theta.
\]
On its first letter $\sigma$, the automaton enters any prefix-initial type
$H=(v,c)$ with $v(p)=1$ iff $p\in\sigma$ for every atomic subformula.
Subsequent transitions are $H\xrightarrow{\sigma}_\theta H'$.
For each future subformula $\gamma=\varphi U^1_I\psi$, define
\[
 G_\gamma=\{H\in\mathsf H_\theta\mid
       c(\gamma)=\bot\text{ or }H\Vdash\psi\}.
\]
The acceptance family consists of all these sets $G_\gamma$.
The distinguished set is $C_\theta=\{H\mid H\Vdash\theta\}$.
\end{definition}

An accepting run visits each $G_\gamma$ infinitely often. If there are no
future subformulas, the acceptance family is empty and every infinite run
is accepting.

\begin{lemma}[Knowledge-free temporal test automata]
\label{lem:kfree-upper}
Let $\theta$ be a knowledge-free KMTL formula over a finite set $\AP'$,
with binary-encoded intervals. The construction above has the following
properties:
\begin{enumerate}[(i)]
\item $|S_\theta|\le2^{\mathsf{poly}(|\theta|)}$, and there are at most
      $|\theta|$ accepting sets. States, initial and transition tests,
      and membership in each accepting set and in $C_\theta$ are
      representable and testable in space polynomial in $|\theta|+|\AP'|$.
\item For every word $w$ and every $i\ge1$, $w,i\models\theta$ iff some
      accepting run of $A_\theta$ on $w$ has its state after position $i$
      in $C_\theta$.
\end{enumerate}
\end{lemma}

\begin{proof}
The representation of a type gives (i). For soundness, let $H_i$ be the
type after position $i$ in an accepting run. We prove by structural
induction that $H_i\Vdash\delta$ iff $w,i\models\delta$ for every
subformula $\delta$ and every position $i$. Atomic and Boolean cases
follow from the letter condition and compositional evaluation. The
next/yesterday equivalences give those two cases, with prefix-initiality
handling yesterday at position $1$.

For $\gamma=\varphi U^1_I\psi$, assume the claim for $\varphi$ and $\psi$.
If a witness exists, the counter is $0$ at that position, and the reverse
recurrence gives the truthful capped distance at every preceding
$\varphi\wedge\neg\psi$-position. If no witness exists, either a first
$\neg\varphi\wedge\neg\psi$-position sets the counter to $\bot$, or
$\varphi\wedge\neg\psi$ holds forever. In the latter case, acceptance
requires a later $\bot$ because $\psi$ never holds. In both cases the
recurrence propagates $\bot$ back to the current position. Thus the
counter has its truthful value, and comparison with $I$ gives the claim.

For a past subformula, induction on the position gives the truthful
counter from prefix-initiality and the forward recurrence. This completes
the structural induction. In particular, $H_i\in C_\theta$ implies
$w,i\models\theta$.

For completeness, assign every Boolean entry its semantic value and every
counter its truthful value. The initial and transition conditions hold by
the definitions of next, yesterday, and witness distance. Each accepting
set is visited infinitely often: from any position, either its future
counter is $\bot$, or a witness position with $\psi$ true is eventually
reached. These types form an accepting run, and its state at position $i$
belongs to $C_\theta$ whenever $w,i\models\theta$.
\end{proof}

This proves Lemma~\ref{lem:kfree-test-automata} of the main text.

\paragraph*{The live predicate}
Fix $\M=(Q,Q_0,\Delta,F_\M)$ and a knowledge-free $\theta$ over $\AP$.
The product of $\M$ and $A_\theta$ has an edge
$(q,s)\longrightarrow(q',s')$ iff some $\sigma\in2^\AP$ satisfies
$(q,\sigma,q')\in\Delta$ and $(s,\sigma,s')\in\delta_\theta$.
A pair $(q,s)$ is in $\mathsf{Live}^\M_\theta$ if $s\in C_\theta$ and
there is an infinite continuation visiting $F_\M\times S_\theta$ and
every set $Q\times G$, $G\in\mathcal F_\theta$, infinitely often.
This is a generalized B\"uchi nonemptiness test in an exponential graph
whose states and edge tests use polynomial space. It is therefore
decidable in polynomial space.

A model run prefix on $\sigma_1\ldots\sigma_i$ ending in $q_i$ extends to
an accepting run whose word satisfies $\theta$ at position $i$ iff
$A_\theta$ has a run on the same prefix ending in some $s$ with
$(q_i,s)\in\mathsf{Live}^\M_\theta$. The forward direction uses the
completeness part of Lemma~\ref{lem:kfree-upper}. The reverse direction
joins the two run prefixes to a live continuation and applies soundness.

\subsection{Observers and Knowledge Tests}

We first set up the notation for abstracting inner knowledge. Consider an occurrence
\[
  \eta=K_a\psi.
\]
Every knowledge operator nested in \(\psi\) is also indexed by \(a\), because
the formula is agent-alternation-free. Let
\(\mathsf{Imm}(\eta)\) be the knowledge-subformula occurrences in \(\psi\) that
are not contained in another knowledge subformula of \(\psi\). Assign a fresh
proposition \(p_\zeta\) to every \(\zeta\in\mathsf{Imm}(\eta)\), and let
\(\widehat\psi\) be obtained by replacing each such \(\zeta\) with
\(p_\zeta\). Then \(\widehat\psi\) is knowledge-free. Write
\[
  \Pin_\eta=\{p_\zeta\mid\zeta\in\mathsf{Imm}(\eta)\},
  \qquad
  \chi_\eta=\neg\widehat\psi,
\]
and let
\[
  A_\eta=(S_\eta,S_{0,\eta},\delta_\eta,\mathcal F_\eta),
  \qquad C_\eta\subseteq S_\eta,
\]
be the temporal test automaton for \(\chi_\eta\). Since \(\chi_\eta\) is knowledge-free
over \(\AP\cup\Pin_\eta\), Lemma~\ref{lem:kfree-test-automata} applies with
\(\AP'=\AP\cup\Pin_\eta\), so \(\delta_\eta\subseteq
S_\eta\times2^{\AP\cup\Pin_\eta}\times S_\eta\) and each letter read by
\(A_\eta\) has the form \(\tau\cup\kappa'\) with \(\tau\in2^{\AP}\) and
\(\kappa'\subseteq\Pin_\eta\). All fresh propositions are chosen
outside \(\AP\) and are distinct across knowledge-subformula occurrences.
For an innermost occurrence, \(O_\eta\) is the observer \(O_{a,\psi}\)
introduced in the main text.

\begin{lemma}[Same-agent invariance]
\label{lem:same-agent-constant}
If \((w,i)\sim_a(w',i)\), then
\[
  w,i\models_\M K_a\xi
  \quad\Longleftrightarrow\quad
  w',i\models_\M K_a\xi.
\]
Consequently, the truth of \(K_a\xi\) at position \(i\) depends only on
\(\obs_a(\pref{w}{i})\).
\end{lemma}

\begin{proof}
Indistinguishable pointed words have the same \(a\)-equivalence class, and
\(K_a\xi\) universally quantifies over that class.
\end{proof}

For every observation prefix \(u=o_1\ldots o_i\), \(i\ge1\), define
\[
  \kappa_\eta(u)=
  \{p_\zeta\in\Pin_\eta\mid
    \forall w\in L(\M),\
    \obs_a(\pref{w}{i})=u\ \Longrightarrow\ w,i\models_\M\zeta\}.
\]
By Lemma~\ref{lem:same-agent-constant}, this is the common inner knowledge
valuation on the observation class. If the class is empty, it is
\(\Pin_\eta\). The definition therefore applies to all observation prefixes.

\begin{definition}[Observer]
\label{def:observer}
For any observation prefix \(u=o_1\ldots o_i\), \(i\ge1\), define
\[
  O_\eta(u)\subseteq Q\times S_\eta
\]
to contain a pair \((q,s)\) iff there are run prefixes
\[
  q_0\xrightarrow{\tau_1}q_1\cdots
  \xrightarrow{\tau_i}q_i
\]
of \(\M\) and
\[
  s_0s_1\ldots s_i
\]
of \(A_\eta\) such that
\[
  q_0\in Q_0,\qquad q_i=q,\qquad s_0\in S_{0,\eta},\qquad s_i=s,
\]
and, for every \(1\le j\le i\),
\[
  \obs_a(\tau_j)=o_j,
  \qquad
  (s_{j-1},\,\tau_j\cup\kappa_\eta(o_1\ldots o_j),\,s_j)
  \in\delta_\eta.
\]
\end{definition}

Thus \(O_\eta(u)\) contains the joint states reachable along all model prefixes
with observation \(u\), while \(A_\eta\) reads the corresponding model letters
extended by the truth values of the immediate inner knowledge subformulas. The model prefixes
in this set need not themselves be live, and liveness is tested separately.

For \(o\in2^{\AP_a}\) and \(\kappa'\subseteq\Pin_\eta\), put
\[
\begin{aligned}
\mathsf{init}_\eta(o,\kappa')
  =\{(q,s)\mid {}&
  \exists q_0\in Q_0,\ \exists\tau\in2^{\AP},\
  \exists s_0\in S_{0,\eta}:\\[-1mm]
& (q_0,\tau,q)\in\Delta,\quad
  \obs_a(\tau)=o,\quad
  (s_0,\tau\cup\kappa',s)\in\delta_\eta\},
\\[1mm]
\mathsf{post}_\eta(O,o,\kappa')
  =\{(q',s')\mid {}&
  \exists(q,s)\in O,\ \exists\tau\in2^{\AP}:\\[-1mm]
& (q,\tau,q')\in\Delta,\quad
  \obs_a(\tau)=o,\quad
  (s,\tau\cup\kappa',s')\in\delta_\eta\}.
\end{aligned}
\]

\begin{lemma}[Observer recursion]
\label{lem:observer-recursion}
For every \(u=o_1\ldots o_i\), \(i\ge1\), and every \(o\in2^{\AP_a}\),
\[
  O_\eta(o_1)=
  \mathsf{init}_\eta(o_1,\kappa_\eta(o_1)),
\]
and
\[
  O_\eta(uo)=
  \mathsf{post}_\eta(O_\eta(u),o,\kappa_\eta(uo)).
\]
\end{lemma}

\begin{proof}
For the first equality, unfold Definition~\ref{def:observer} with \(i=1\).
For the second, split a joint run over \(uo\) after its prefix over \(u\).
The last model transition reads some \(\tau\) with \(\obs_a(\tau)=o\), and the
last test-automaton transition reads
\(\tau\cup\kappa_\eta(uo)\). The converse concatenation gives every pair in
the stated post-image.
\end{proof}

\paragraph*{Relative liveness}
The observer \(O_\eta(u)\) records the prefix up to the current position.
Along a future counterexample continuation, the test automaton \(A_\eta\) must also
receive the values of the propositions in \(\Pin_\eta\). Those values depend
on the continued observation history. We therefore retain the observers of
all knowledge subformulas properly nested in \(\eta\).

Let \(\mathsf{Desc}(\eta)\) be the set of all knowledge-subformula occurrences
strictly inside \(\eta\). For any observation prefix \(u\), define the
inner tuple
\[
  \mathcal I_\eta(u)=
  \bigl(O_\zeta(u)\bigr)_{\zeta\in\mathsf{Desc}(\eta)}.
\]
Here every \(\zeta\) has agent \(a\), so all components are indexed by the same
observation prefix. The definitions below proceed by induction on the
knowledge depth of \(\eta\).

For an abstract inner tuple \(\mathcal I=(O_\zeta)_\zeta\), define the valuation
of the immediate inner propositions by
\[
  p_\zeta\in\lambda_\eta(\mathcal I)
\]
iff no pair \((\tilde q,\tilde s)\in O_\zeta\) satisfies
\[
  \mathsf{Live}^\M_\zeta
  (\tilde q,\tilde s,\mathcal I_\zeta),
\]
where \(\mathcal I_\zeta\) is the subtuple indexed by
\(\mathsf{Desc}(\zeta)\). The predicate on the right concerns a subformula of
strictly smaller knowledge depth.

For \(o\in2^{\AP_a}\), define \(\mathsf{post}(\mathcal I,o)\) by updating its
components from the innermost occurrences outwards. When the component \(O_\zeta\) is
updated, all components strictly inside \(\zeta\) have already been updated and
determine the new valuation \(\lambda_\zeta(\mathcal I'_\zeta)\), where
\(\mathcal I'_\zeta\) is the updated descendant subtuple. The new component is
\[
  \mathsf{post}_\zeta(O_\zeta,o,\lambda_\zeta(\mathcal I'_\zeta)).
\]

Define
\[
  \mathsf{Live}^\M_\eta(q,s,\mathcal I)
\]
on configurations \((q,s,\mathcal I)\), where \(q\in Q\), \(s\in S_\eta\),
and \(\mathcal I\) is an inner tuple for \(\eta\). There is a transition
\[
  (q,s,\mathcal I)\longrightarrow(q',s',\mathcal I')
\]
iff some \(\tau\in2^{\AP}\) satisfies
\[
  (q,\tau,q')\in\Delta,
  \qquad
  \mathcal I'=\mathsf{post}(\mathcal I,\obs_a(\tau)),
\]
and
\[
  (s,\tau\cup\lambda_\eta(\mathcal I'),s')\in\delta_\eta.
\]
The predicate holds if \(s\in C_\eta\) and this graph has an infinite
continuation visiting \(F_\M\) in the first component and every
\(G\in\mathcal F_\eta\) in the second component infinitely often.

If \(\eta\) has no inner knowledge subformula, the tuple is empty and the definition
is the ordinary live predicate for \(\chi_\eta=\neg\psi\). Otherwise,
successor generation invokes liveness predicates only for proper descendants.
For total input length \(n\), fix a uniform bound \(W=2^{p(n)}\), for a
polynomial \(p\), on the space for a relative-liveness configuration and its
local update buffers. Each graph has at most \(2^W\) vertices. Generalized
B\"uchi nonemptiness reduces to sequential reachability tests. Savitch's
midpoint procedure stores \(O(W^2)\) bits and calls tests for proper
descendants only when evaluating an edge. Let \(R_d\) bound the space at
nesting depth \(d\), with innermost tests at depth~\(0\). For a uniform
constant \(c\), we obtain
\[
  R_0\le cW^2,
  \qquad
  R_d\le cW^2+\max_{j<d}R_j\quad(d\ge1).
\]
The caller's stack remains stored during a descendant test, whose workspace
is reused across sequential calls. Thus, \(R_d=O((d+1)W^2)\), which is
exponential in \(n\) since \(d\le n\). All tests are deterministic, so their
answers can be negated in the definition of \(\lambda_\eta\).

\begin{lemma}[Knowledge test]
\label{lem:knowledge-test}
Let \(\eta=K_a\psi\) and \(u=o_1\ldots o_i\), \(i\ge1\). Then
\[
 \begin{aligned}
 &\exists(q,s)\in O_\eta(u):\
   \mathsf{Live}^\M_\eta(q,s,\mathcal I_\eta(u))\\
 &\qquad\Longleftrightarrow\
   \exists w'\in L(\M):\
   \obs_a(\pref{w'}{i})=u\ \text{ and }\ w',i\models_\M\neg\psi.
 \end{aligned}
\]
Consequently, if \(w\in L(\M)\) has observation prefix \(u\), then
\(w,i\models_\M\eta\) iff \(O_\eta(u)\) contains no such live pair.
\end{lemma}

\begin{proof}
We prove the equivalence together with correctness of the computed knowledge
values, by induction on the knowledge depth of \(\eta\), for all observation
prefixes. The induction hypothesis applies to every proper descendant.
The observer recursion then gives the semantic observer sets by induction
on prefix length, and their knowledge tests give
\(\lambda_\eta(\mathcal I_\eta(u))=\kappa_\eta(u)\).
Thus \(A_\eta\) reads the correct inner valuations on every continuation.

Suppose first that \((q,s)\in O_\eta(u)\) satisfies the relative live
predicate. The observer supplies a model prefix with observation \(u\) and a
compatible \(A_\eta\)-prefix ending in \((q,s)\). The live continuation extends
both prefixes to accepting runs. Let \(w'\in L(\M)\) be the resulting model word and let
\(\widetilde w'\) be its extension by the fresh valuations. Since
\(s\in C_\eta\), soundness of Lemma~\ref{lem:kfree-upper} gives
\[
  \widetilde w',i\models\chi_\eta.
\]
By the definition of the fresh valuations, this is equivalent to
\(w',i\models_\M\neg\psi\), proving the forward direction.

Conversely, let \(w'\in L(\M)\) satisfy
\[
  \obs_a(\pref{w'}{i})=u
  \quad\text{and}\quad
  w',i\models_\M\neg\psi.
\]
Extend \(w'\) at every position with the truth values of the immediate inner
knowledge subformulas. The resulting word satisfies \(\chi_\eta\) at position \(i\).
Completeness of Lemma~\ref{lem:kfree-upper} yields an accepting run of
\(A_\eta\) whose state at \(i\) lies in \(C_\eta\). Its prefix, together with
an accepting run of \(\M\) over \(w'\), shows that the corresponding pair lies
in \(O_\eta(u)\). The suffixes of the two runs, with the descendant observers
updated along \(w'\), witness the relative live predicate.
The equivalence proves correctness of the computed knowledge value for
\(\eta\). When the observation class is empty, there is no live pair, and
the computed value is true, as required by the universal definition.
\end{proof}

\subsection{Alternation Depth One and the Global Product}

Let \(\mathsf{Know}(\varphi)\) be the set of knowledge-subformula occurrences
of \(\varphi\). For each such occurrence, construct its observer and relative
live predicate as above. Because \(\varphi\in\KMTLone\), every descendant of an
occurrence \(K_a\psi\) has the same agent \(a\). Different agents may occur in
separate, non-nested knowledge subformulas.

Replace every maximal knowledge subformula of \(\neg\varphi\) by its fresh
proposition, obtaining a knowledge-free formula
\(\theta_{\mathsf{top}}\). The observers of all nested descendants are retained
because they are needed to evaluate the maximal knowledge subformulas.
The alphabet of \(A_{\theta_{\mathsf{top}}}\) includes \(\AP\) and the fresh
propositions for all knowledge occurrences.

A global configuration is
\[
  (q,t,\mathcal O),
\]
where \(q\in Q\), \(t\in S_{\theta_{\mathsf{top}}}\), and
\[
  \mathcal O=(O_\eta)_{\eta\in\mathsf{Know}(\varphi)}
\]
contains one observer set per knowledge occurrence. The tuple determines a valuation
\(\lambda_{\mathcal O}\) of the fresh propositions for all knowledge
occurrences: the proposition for \(\eta\) is true iff its observer
contains no pair satisfying the relative live predicate from the corresponding
descendant subtuple.

An initial configuration is obtained by choosing
\[
  (q_0,\sigma,q)\in\Delta
  \quad\text{with}\quad q_0\in Q_0.
\]
Initialize all observer sets on \(\sigma\), from the innermost occurrences outwards.
Each occurrence uses the observation \(\obs_a(\sigma)\) of its own agent. This yields
\(\mathcal O\) and \(\lambda_{\mathcal O}\). Choose
\(t_0\in S_{0,\theta_{\mathsf{top}}}\) and
\(t\in C_{\theta_{\mathsf{top}}}\) such that
\[
  (t_0,\sigma\cup\lambda_{\mathcal O},t)
  \in\delta_{\theta_{\mathsf{top}}}.
\]

From \((q,t,\mathcal O)\), choose a model transition
\[
  (q,\sigma',q')\in\Delta.
\]
Update every observer on the observation of \(\sigma'\) for its agent, again
from the innermost occurrences outwards, to obtain \(\mathcal O'\). Compute
\(\lambda_{\mathcal O'}\), and choose \(t'\) with
\[
  (t,\sigma'\cup\lambda_{\mathcal O'},t')
  \in\delta_{\theta_{\mathsf{top}}}.
\]
The successor is \((q',t',\mathcal O')\). The acceptance condition requires
the first component to visit \(F_\M\) infinitely often and the second to visit
every \(G\in\mathcal F_{\theta_{\mathsf{top}}}\) infinitely often. Acceptance conditions of the
inner test automata are checked within the relative live predicates and impose no
additional condition on the global path.

\begin{lemma}[Global correctness]
\label{lem:global-correct}
The global product has an accepting run iff there is a word \(w\in L(\M)\)
such that
\[
  w,1\models_\M\neg\varphi.
\]
\end{lemma}

\begin{proof}
Fix \(w\in L(\M)\). Along its prefixes, joint induction on knowledge-nesting depth establishes
both the observer invariant and correctness of the computed valuations.
For an innermost occurrence, Lemma~\ref{lem:observer-recursion} gives the
observer invariant and Lemma~\ref{lem:knowledge-test} gives its knowledge
value. For an outer occurrence, the induction hypothesis supplies the
inner valuations needed by the update. The same two lemmas therefore give
both claims at that occurrence. In particular, after position \(i\) the
observer component for \(\eta=K_a\psi\) is
\[
  O_\eta(\obs_a(\pref{w}{i})).
\]
Lemma~\ref{lem:knowledge-test} therefore assigns every fresh proposition the
truth value of the knowledge subformula it represents. Structural induction on
the formula shows that the extended word satisfies
\(\theta_{\mathsf{top}}\) at a position iff \(w\) satisfies
\(\neg\varphi\) there.

If the global product has an accepting run, its first component is an accepting
run of \(\M\) over some \(w\in L(\M)\). Its second component is an accepting
run of \(A_{\theta_{\mathsf{top}}}\), and the initial global configuration
ensures that its state at position \(1\) belongs to
\(C_{\theta_{\mathsf{top}}}\). Soundness of
Lemma~\ref{lem:kfree-upper} gives
\[
  w,1\models_\M\neg\varphi.
\]

Conversely, suppose \(w\in L(\M)\) satisfies
\(w,1\models_\M\neg\varphi\), and fix an accepting run of \(\M\) over \(w\).
The observer recursion determines \(\mathcal O\) at every position, and the
knowledge test gives the correct fresh valuation. Completeness of
Lemma~\ref{lem:kfree-upper} supplies an accepting run of
\(A_{\theta_{\mathsf{top}}}\) whose state at position \(1\) is distinguished. These
components form an accepting run of the global product.
\end{proof}

\end{document}